\documentclass[11pt]{article}
\usepackage{amsmath,amsfonts,amssymb,amsthm}
\usepackage{fullpage}
\usepackage[ruled]{algorithm2e}

\usepackage{subfigure} \usepackage{epstopdf}
\usepackage{graphicx}
\usepackage{cite}
\usepackage{color}
\usepackage[sans]{dsfont}
\usepackage{mathtools}
\usepackage{hyperref}
\usepackage{cleveref}
\usepackage{tcolorbox}
\usepackage{multirow}
\usepackage{xspace}
\tcbuselibrary{skins,breakable}
\tcbset{enhanced jigsaw}

\newtheorem{theorem}{Theorem}

\newtheorem{lemma}[theorem]{Lemma}
\newtheorem{observation}[theorem]{Observation}

\newtheorem{claim}[theorem]{Claim}

\newcommand{\set}[1]{\left\{ #1 \right\}}

\newcommand{\dset}{{\mathcal{D}}}

\newcommand{\qset}{{\mathcal{Q}}}

\newcommand{\lset}{{\mathcal{L}}}

\newcommand{\cset}{{\mathcal{C}}}

\newcommand{\gset}{{\mathcal{G}}}

\newcommand{\rset}{{\mathcal{R}}}

\newcommand{\D}{{\dset}}
\newcommand{\wl}{W_{\mathsf{L}}}
\newcommand{\wri}{W_{\mathsf{R}}}

\newcommand{\dist}{\textnormal{\textsf{dist}}}
\newcommand{\cost}{\textnormal{\textsf{cost}}}

\newcommand\vG{G}

\newcounter{note}

\newenvironment{properties}[2][0]
{
	\begin{enumerate} \setcounter{enumi}{#1}}{\end{enumerate}}

\begin{document}
	
%%%%% For final version, uncomment author info and acknowledgement %%%%%%%%%%%%%

\begin{titlepage}
	
	\title{Paths and Intersections: Characterization of \\ Quasi-metrics in Directed Okamura-Seymour Instances}

%\iffalse	
	\author{Yu Chen\thanks{National University of Singapore, Singapore. Email: {\tt yu.chen@nus.edu.sg}. Work done while the author was a postdoc at EPFL.} \and Zihan Tan\thanks{Rutgers University, NJ, USA. Email: {\tt zihantan1993@gmail.com}.}} 
%\fi

	\maketitle

	\thispagestyle{empty}
	\begin{abstract}
We study the following distance realization problem. Given a quasi-metric $D$ on a set $T$ of terminals, does there exist a directed Okamura-Seymour graph that realizes $D$ as the (directed) shortest-path distance metric on $T$? 
We show that, if we are further given the circular ordering of terminals lying on the boundary, then Monge property is a sufficient and necessary condition.
This generalizes previous results for undirected Okamura-Seymour instances. 

With the circular ordering, we give a greedy algorithm for constructing a directed Okamura-Seymour instance that realizes the input quasi-metric.
The algorithm takes the dual perspective concerning flows and routings, and is based on a new way of analyzing graph structures, by viewing graphs as \emph{paths and their intersections}. We believe this new understanding is of independent interest and will prove useful in other problems in graph theory and graph algorithms.

We also design an efficient algorithm for finding such a circular ordering that makes $D$ satisfy Monge property, if one exists. Combined with our result above, this gives an efficient algorithm for the distance realization problem.

	\end{abstract}
\end{titlepage}

\tableofcontents

\newpage

\section{Introduction}

%{\color{red} quasi-metrics}

We study the following distance realization problem: given a quasi-metric $D$ on a set $T$ of \emph{terminals}, does there exist a directed Okamura-Seymour graph that realizes $D$ as the (directed) shortest-path distance metric on $T$? Specifically, a function $D: T\times T\to \mathbb{R}^+$ is a \emph{quasi-metric}\footnote{Sometimes we allow $D$ to take value $+\infty$ to better reflect the property of directed shortest-path distance metrics, that is $D: T\times T\to \mathbb{R}^+\cup \set{+\infty}$; and the triangle inequality property needs to be naturally appended as follows: for all $t,t',t''\in T$, if $D(t,t'')=+\infty$, then either $D(t,t')$ or $D(t',t'')$ is $D(t,t'')=+\infty$.}, iff
\begin{itemize}
\item for all $t,t'\in T$, $D(t,t')\ge 0$, and $D(t,t')= 0$ iff $t=t'$; and
\item for all $t,t',t''\in T$, $D(t,t'')\le D(t,t')+D(t',t'')$.
\end{itemize}
Quasi-metrics do not have the symmetry axiom (i.e., it is not required that $D(t,t')=D(t',t)$ holds).

A graph $G$ is called a \emph{directed Okamura-Seymour graph} with respect to $T$, iff $G$ is a planar graph with $T\subseteq V(G)$, and all terminals in $T$ lie on the boundary of the outer face in its plane embedding. 
We say that \emph{$G$ realizes $D$} as the directed shortest-path distance metric on $T$, iff for every ordered terminal pair $(t,t')$, the minimum length of any directed path in $G$ going from $t$ to $t'$, denoted by $\dist_{\vG}(t,t')$, equals $D(t,t')$.

The problem is interesting from three aspects. First of all, the problem itself is a typical \emph{distance realization} problem. In a problem of this type, we are given a family $\gset$ of graphs and a metric $D$ on a set $T$ of vertices designated as \emph{terminals}. The goal is to decide whether or not there exists a graph $G\in \gset$ and a way of identifying some of its vertices as terminals in $T$, such that for every pair $t,t'\in T$ of terminals, $\dist_{G}(t,t')=D(t,t')$.
Apart from being a fundamental problem in metric graph theory, distance realization problems have also found numerous applications in computational biology (e.g., constructing phylogenetic trees) \cite{dress2012basic}, chemistry \cite{janezic2015graph}, hierarchical classification problems \cite{gordon1987review}, and network tomography \cite{chung2001distance}.

The most extensively studied cases are when $\gset$ is the family of 
(i) ultrametrics (related to hierarchical clasifications) \cite{gordon1987review};
(ii) trees \cite{hakimi1965distance,pereira1969note,buneman1974note,dress1984trees}, where a four-vertex condition is proved to be sufficient and necessary; (iii) cactus graphs \cite{hayamizu2020recognizing}; and (iv) undirected Okamura-Seymour instances \cite{ChangO20,hurkens1988tidy}, all for undirected graph families. 
In this paper, we study the problem for the first directed graph family: the family of directed Okamura-Seymour instances.

Second, a sufficient and necessary condition for the distance realization problem on a graph family $\gset$ usually has a concrete connection with the multi-commodity flow multi-cut gap of $\gset$, a quantity related to the approximation ratios for many classic flow/cut based graph problems. 
For example, for undirected graphs, the flow-cut gap for Okamura-Seymour instances is $1$ \cite{okamura1981multicommodity}, and the Monge property is shown to be sufficient and necessary for the corresponding distance realization problem \cite{ChangO20,hurkens1988tidy}. The two proofs are in fact equivalent to each other, as illustrated in Section 74.1e of \cite{schrijver2003combinatorial}.
On the other hand, the flow-cut gap for directed graphs remain quite elusive. It is shown to be $\Omega(n^{1/7})$ for $n$-vertex general graphs \cite{chuzhoy2009polynomial}, which is in sharp contrast with the $\Theta(\log n)$ gap for undirected graphs \cite{leighton1999multicommodity,linial1995geometry}. Recently, Kawarabayashi and Sidiropoulos showed that the flow-cut gap for directed planar graphs is $O(\log^3 n)$, and Sidiropoulos mentioned\footnote{at the Advances on Metric Embedding Workshop at FOCS 22, video recording can be accessed \url{https://vimeo.com/user39621409/review/771333384/dfe1ed17aa}} that determining the flow-cut gap for directed Okamura-Seymour instances remains an interesting open problem.

Third, this problem also fits into a recent line of works studying directed shortest-path structures and distance metrics \cite{bodwin2019structure,akmal2022local,cizma2022geodesic,cizma2023irreducible}, with the motivation being that, while researchers have developed powerful tools for exploring the shortest-path distance metrics of undirected graphs, our knowledge for directed graphs remains quite limited. In this sense, our work can be viewed as a basic step towards a better understanding of the shortest-path distance metrics of directed graphs, and we hope that our work could inspire more research in this direction.

\subsection{Our Results}

Our main result, summarized in \Cref{thm: directed 4 point condition}, says that, if we are further given the circular ordering of the terminals lying on the boundary, then Monge property is a sufficient and necessary condition for the input quasi-metric to be realizable by a directed Okamura-Seymour instance. %In order to utilize Monge property, we need a circular $\sigma$ ordering on terminals in $T$. 
%We say a quasi-metric $D$ is an \emph{Okamura-Seymour quasi-metric} on $T$ with respect to a circular ordering $\sigma$ on $T$, iff there exists an Okamura-Seymour directed graph $G$ where the terminals lie on the boundary of the outer face in the circular order $\sigma$.
%Our main result can be stated as the following theorem.

\begin{theorem}
\label{thm: directed 4 point condition}
A quasi-metric $D$ on $T$ is realizable by a directed Okamura-Seymour instance where the terminals in $T$ lie on the boundary in the circular ordering $\sigma$, iff for all quadruples $t_1,t_2,t_3,t_4$ of terminals in $T$ appearing in $\sigma$ in this order, $D(t_1,t_3)+D(t_2,t_4)\ge D(t_1,t_4)+D(t_2,t_3)$. Moreover, such a quasi-metric can be realized by a directed Okamura-Seymour instance on $O(|T|^6)$ vertices.
\end{theorem}

Our \Cref{thm: directed 4 point condition} generalizes the previous result on undirected Okamura-Seymour instances \cite{ChangO20,hurkens1988tidy} to the directed case. The proof requires fundamentally new ideas, and is necessarily much more complicated than the undirected case. We provide in \Cref{subsec: overview} an overview of our approach, discussing the barrier of the previous approaches in the directed case and our way of overcoming it.
Our approach proves useful a new understanding of graphs: instead of viewing a graph as consisting of vertices and edges, we view it as being formed by \emph{paths and their intersections}. We believe this new understanding is of independent interest, and will facilitate other research in graph algorithms and graph theory.

\textbf{Remark.} \Cref{thm: directed 4 point condition} also holds if we allow the quasi-metric $D$ to take value $+\infty$. The four-term inequality needs to be naturally augmented by the following condition: for all quadruples $t_1,t_2,t_3,t_4$ of terminals in $T$ appearing in $\sigma$ in this order, if one of $D(t_1,t_4),D(t_2,t_3)$ is $+\infty$, then at least one of $D(t_1,t_3),D(t_2,t_4)$ is $+\infty$.

With the sufficient and necessary conditions in \Cref{thm: directed 4 point condition}, we show that we can efficiently decide whether or not there exists a directed Okamura-Seymour instance $G$ realizing the input quasi-metric $D$, by efficiently finding a good circular ordering $\sigma$ or certifying that none exists.

\begin{theorem}
\label{cor: algorithm}
There is an efficient algorithm, that, given a quasi-metric $D$, correctly decides if there exists a directed Okamura-Seymour instance $G$ that realizes $D$.
\end{theorem}

\paragraph{Related Work.} 
There are other works on distance realization problems besides those mentioned in the introduction, e.g.,
\cite{patrinos1972distance,simoes1982submatrices,simoes1987graph,bandelt1990recognition,varone1998trees,bar2022graph,bar2023composed}. See the survey by Aouchiche and Hansen \cite{aouchiche2014distance}.

The flow-cut gap in directed Okamura Seymour instances is proved to be $1$ in some special cases, for example for the demands when the sources and the sinks appear on the boundary in a non-crossing manner \cite{diaz1972multicommodity}, and when the graph is acyclic \cite{nagamochi1989max}.
There are also other interesting works on compressing Okamura-Seymour metrics \cite{goranci2017improved,chang2018near,li2019planar,mozes2022improved}.

\subsection{Technical Overview}
\label{subsec: overview}

We provide a high-level overview on our approaches in the proof of \Cref{thm: directed 4 point condition}.

\subsubsection*{Comparison between undirected and directed cases}

Previous results \cite{ChangO20,hurkens1988tidy} showed that, given a metric $D$ and a circular ordering $\sigma$, Monge property with respect to $\sigma$ is a sufficient and necessary condition for the metric $D$ to be realizable by an undirected Okamura-Seymour instance, and more surprisingly, a universal Okamura-Seymour instance. Specifically, there exist
\begin{itemize}
\item a planar graph $G^*$;
\item a designated face $F$ in its plane embedding; and
\item for every pair $u,u'$ of vertices on $F$, a designated path $P_{u,u'}$ in $G$;
\end{itemize} 
such that, for every $\sigma$ and metric $D$ satisfying Monge property with respect to $\sigma$, there are
\begin{itemize}
\item a way of identifying terminals in $T$ with vertices on $F$, respecting the ordering $\sigma$; and
\item a way of choosing edge weights $\set{w(e)}_{e\in E(G^*)}$, such that for every (unordered) pair $t,t'\in T$, $P_{t,t'}$ is the shortest $t$-$t'$ path in $G^*$, and its length is $D(t,t')$. 
\end{itemize}
In order words, the skeleton of the graph $G^*$ and its shortest path structure are completely fixed, and upon receiving a metric $D$, we just need to adjust its edge weights to realize $D$ via this structure.

Previous work give different constructions of $G^*$. Section 74.1e of \cite{schrijver2003combinatorial} presents a nest-type construction, by drawing a straight line between every pair of boundary vertices and then replace each intersection with a vertex. \cite{ChangO20} gave a more compact structure, using a quarter-grid construction.
%See XXX for an illustration on $6$ terminals.

The directed case, however, does not admits such a universal construction, for the following reasons.

Denote $T=\set{a,b,t_1,t_2,t_3,t_4}$,
let $\sigma=(a,t_4,t_1,b,t_3,t_2)$ be
a circular ordering on $T$ (represented in the clockwise way), and let $D$ be a quasi-metric on $T$, 
where 
\begin{itemize}
\item for every pair $t,t'\in T$ appearing consecutively in $\sigma$, $D(t,t')=4$, $D(t',t)=1$;
\item for every tuple $t,t'',t'\in T$ appearing consecutively in $\sigma$, $D(t,t')=4$, $D(t',t)=2$;
\item for every antipodal pair $(t,t')\in \set{(a,b),(t_1,t_2),(t_3,t_4)}$, $D(t,t')=D(t',t)=3$;
\end{itemize}
It is easy to verify that the quasi-metric $D$ satisfies Monge property with respect to $\sigma$.

Imagine that we are constructing such a universal graph $G^*$ for the directed case. We have inserted a directed path $\pi_{t_1,t_2}$ from $t_1$ to $t_2$ designated to be the shortest $t_1$-$t_2$ path, and similarly a directed path $\pi_{t_3,t_4}$. According to $\sigma$, these two paths have to create an intersection, which we denote by $p$. Now we want to insert the $a$-$b$ shortest path $\pi_{a,b}$, and we are faced with the following question:

\emph{Should the path $\pi_{a,b}$ pass around $p$ from its left or from its right?}

In fact, if we blindly let path $\pi_{a,b}$ pass around $p$ from its left, as shown in \Cref{fig: routing2}, then we are in trouble. This is since, on the one hand, from our definition of the quasi-metric $D$,
\[D(a,b)+D(t_1,t_2)+D(t_3,t_4)=3+3+3<4+4+4=D(a,t_4)+D(t_1,b)+D(t_3,t_2);\]
on the other hand, since path $\pi_{a,b}$ inevitably causes intersections with paths $\pi_{t_1,t_2},\pi_{t_3,t_4}$, as shown in \Cref{fig: routing3}, we can extract a $t_1$-$b$ path, an $a$-$t_4$ path, and a $t_3$-$t_2$ path, by non-repetitively using the edges from the shortest paths $\pi_{t_1,t_2},\pi_{t_3,t_4}$, and $\pi_{a,b}$, which implies that 
\[D(a,b)+D(t_1,t_2)+D(t_3,t_4)\ge D(a,t_4)+D(t_1,b)+D(t_3,t_2),\]
a contradiction. 

\begin{figure}[h!]
	\centering
	\subfigure[Paths $\pi_{t_1,t_2}$ and $\pi_{t_3,t_4}$.]
	{\scalebox{0.08}{\includegraphics{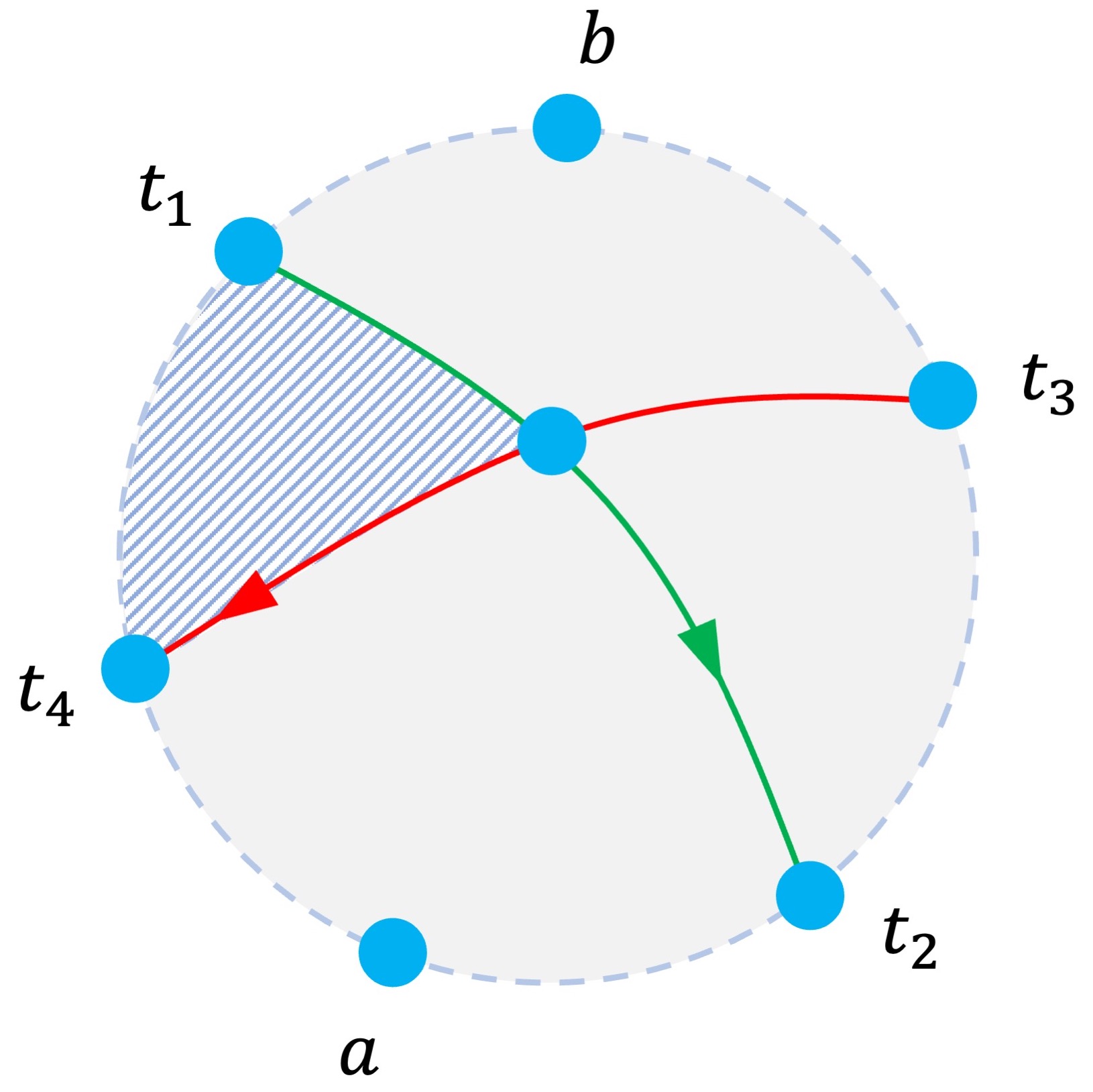}}\label{fig: routing1}}
	\hspace{0.5cm}
	\subfigure[A way of inserting $\pi_{a,b}$.]
	{
		\scalebox{0.08}{\includegraphics{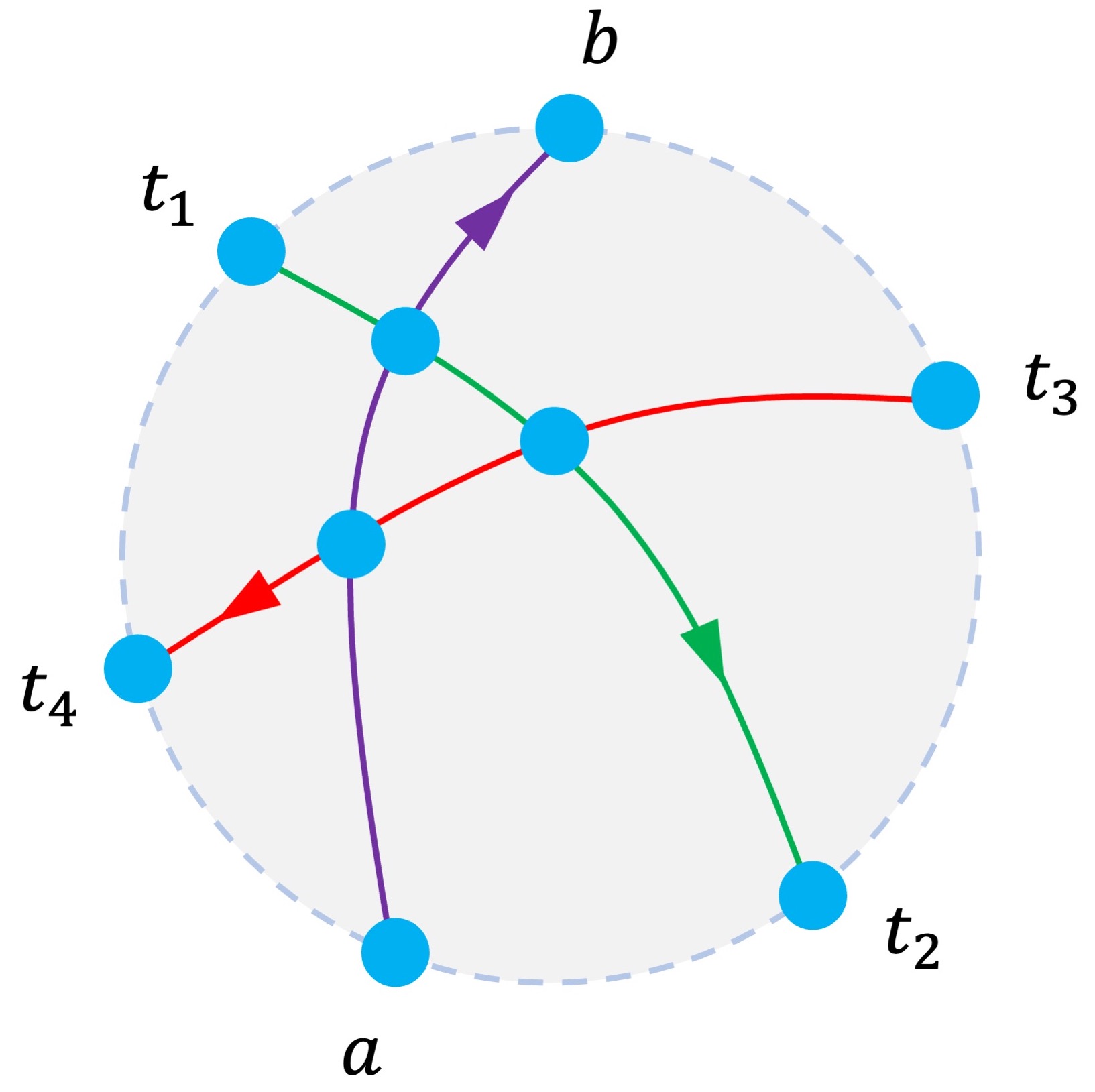}}\label{fig: routing2}}
	\hspace{0.5cm}
	\subfigure[Extracting a $t_1$-$b$ path, an $a$-$t_4$ path, and a $t_3$-$t_2$ path.]
	{
		\scalebox{0.08}{\includegraphics{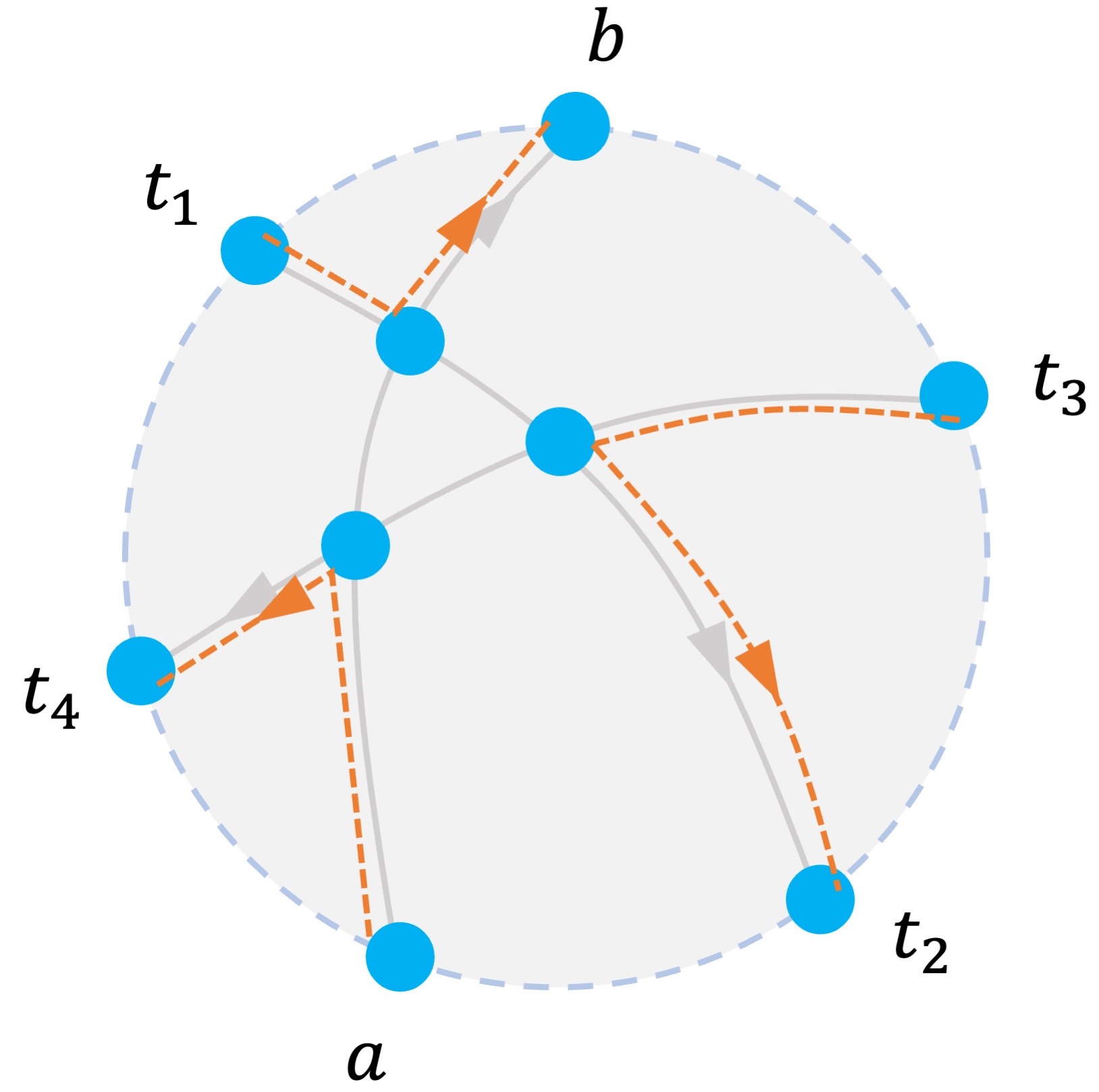}}\label{fig: routing3}}
	\caption{An illustration of an incorrect way of inserting path $\pi_{a,b}$ under certain $(D,\sigma)$.\label{fig: routing}}
\end{figure}

Therefore, under the current quasi-metric $D$, we must let $\pi_{a,b}$ pass around $p$ from its right.
However, if we mirror flip the quasi-metric by simply changing $\sigma$ to be its mirror flip $(t_2,t_3,b,t_1,t_4,a)$, then we have to let $\pi_{a,b}$ pass around $p$ from its left, by the same arguments. This means that the relative position between $\pi_{a,b}$ and $p$, which is essentially what constitutes the shortest path structure, has to depend on $D$ and $\sigma$, and therefore cannot be universal in the directed setting.

\subsubsection*{Forbidden areas, and getting around them}

Continuing with the example above, we can in fact show that the path $\pi_{a,b}$ is not allowed to touch the blue area shown in \Cref{fig: routing1}. This area can be understood as a ``forbidden area'' caused by the pair $\set{D(a,b), D(t_1,t_2), D(t_3,t_4)},  \set{D(a,t_4), D(t_1,b), D(t_3,t_2)}$ of sets of values in $D$. We call such a pair of sets a \emph{restricting pair}.

In order to construct a directed Okamura-Seymour instance realizing $D$, we use a greedy algorithm that inserts terminal shortest paths one by one.
It suffices to show that, upon receiving every pair $(a,b)$ of terminals, there exists a way of inserting the directed path $\pi_{a,b}$ into the current graph that gets around all forbidden areas. Equivalently, we need to show that the union of all forbidden areas does not block the way from $a$ to $b$ inside the disc.

Note that the disc boundary is separated by $a$ and $b$ into two segments: the right segment (from $b$ clockwise to $a$) and the left segment (from $a$ clockwise to $b$).
Typically, a forbidden area rises from either the left segment or the right segment and extends to somewhere in the middle. Therefore, in a ``geometry duality'' sense,
\emph{the forbidden areas block the way from $a$ to $b$ iff there exist a forbidden area from the left segment and a forbidden area from the right segment that intersect each other}.

But how can we ensure that there are no such intersecting forbidden areas? As an example, how do we rule out the possibility shown in \Cref{fig: forbidden}, where a forbidden area $p$-$t_1$-$t_4$ from the left segment intersects the forbidden area $q$-$t_6$-$t_7$ from the right segment?

\begin{figure}[h!]
	\centering
	\includegraphics[scale=0.12]{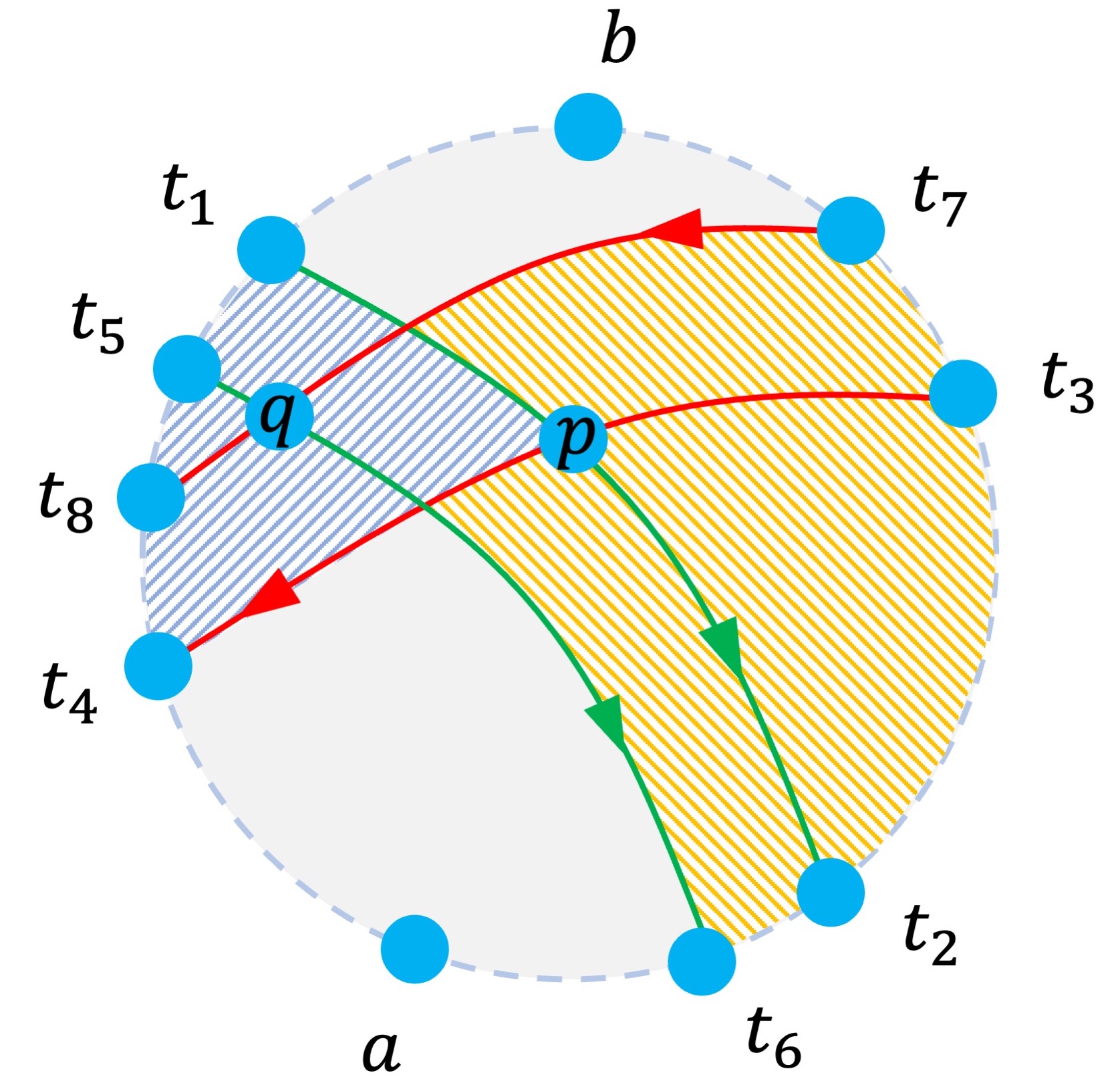}
	\caption{An illustration of intersecting forbidden areas (yellow is partially covered by blue).\label{fig: forbidden}}
\end{figure}

The answer is to use the Monge property of the input quasi-metric $D$ (this is the first and the only time we use it). Specifically, the forbidden area $p$-$t_1$-$t_4$ is caused by the restricting pair 
\begin{itemize}
	\item $\cset_1=\set{(a,b),(t_1,t_2),(t_3,t_4)}$; and
	\item $\cset'_1=\set{(a,t_4),(t_1,b),(t_3,t_2)}$;
\end{itemize}
and the forbidden area  $q$-$t_6$-$t_7$ is caused by the restricting pair 
\begin{itemize}
	\item $\cset_2=\set{(a,b),(t_5,t_6),(t_7,t_8)}$; and
	\item $\cset'_2=\set{(a,t_6),(t_5,t_8),(t_7,t_b)}$.
\end{itemize}
Denote $D(\cset_1)=D(a,b)+D(t_1,t_2)+D(t_3,t_4)$, and define $D(\cset'_1),D(\cset_2),D(\cset'_2)$ similarly. For $(\cset_1,\cset'_1)$ and $(\cset_2,\cset'_2)$ to be restricting pairs, $D(\cset_1)<D(\cset'_1)$ and $D(\cset_2)<D(\cset'_2)$ must hold, and so $D(\cset_1)+D(\cset_2)<D(\cset'_1)+D(\cset'_2)$, i.e.,
\begin{equation}
\label{eqn_1}
D\bigg(2(a,b),(t_1,t_2),(t_3,t_4),(t_5,t_6),(t_7,t_8)\bigg)<D\bigg((a,t_4),(t_1,b),(t_3,t_2),(a,t_6),(t_5,t_8),(t_7,b)\bigg).
\end{equation}
However, this is impossible, since from the Monge property of $D$,
\begin{itemize}
	\item $D(t_1,t_2)+D(t_7,t_8)\ge {\color{red}D(t_1,t_8)}+{\color{blue}D(t_7,t_2)}$, as $(t_1,t_2)$ and $(t_7,t_8)$ cross; and
	\item $D(a,b)+D(t_3,t_4)\ge D(a,t_4)+{\color{brown}D(t_3,b)}$, as $(a,b)$ and $(t_3,t_4)$ cross; and
	\item ${\color{brown}D(t_3,b)}+{\color{blue}D(t_7,t_2)}\ge D(t_3,t_2)+D(t_7,b)$, as $(t_3,b)$ and $(t_7,t_2)$ cross; and
	\item $D(a,b)+D(t_5,t_6)\ge D(a,t_6)+{\color{brown}D(t_5,b)}$, as $(a,b)$ and $(t_5,t_6)$ cross; and
	\item ${\color{brown}D(t_5,b)}+{\color{red}D(t_1,t_8)}\ge D(t_1,b)+D(t_5,t_8)$, as $(t_5,b)$ and $(t_1,t_8)$ cross.
\end{itemize}
Combining all these inequalities, we get that 
\[
D\bigg(2(a,b),(t_1,t_2),(t_3,t_4),(t_5,t_6),(t_7,t_8)\bigg)\ge D\bigg((a,t_4),(t_1,b),(t_3,t_2),(a,t_6),(t_5,t_8),(t_7,t_b)\bigg).
\]
a contradiction to Inequality~\ref{eqn_1}.

To summarize, we refute the possibility of intersecting forbidden areas in \Cref{fig: forbidden} by showing that their corresponding pairs $(\cset_1,\cset'_1),(\cset_2,\cset'_2)$ cannot both be restricting pairs. In other words, although there are Monge quasi-metrics $D$ (quasi-metrics $D$ satisfying Monge property) allowing $(\cset_1,\cset'_1)$ to be a restricting pair, and there are Monge quasi-metrics $D$ allowing $(\cset_2,\cset'_2)$ to be a restricting pair, there are no Monge quasi-metrics $D$ allowing both $(\cset_1,\cset'_1)$ and $(\cset_2,\cset'_2)$ to be restricting pairs.

The core in the analysis above is the very first inequality $D(t_1,t_2)+D(t_7,t_8)\ge {\color{red}D(t_1,t_8)}+{\color{blue}D(t_7,t_2)}$, which is the bridge between the collections $\cset_1,\cset_2$: since $(t_1,t_2)$ as a terminal pair in $\cset_1$ and $(t_7,t_8)$ as a terminal pair in $\cset_2$ cross, we are able to put them together on the LHS of a Monge inequality, and thereby combine the sets $\cset_1,\cset_2$ to show that in $D$ they cannot be simultaneously smaller than $\cset'_1,\cset'_2$, respectively.

This example captures the main idea of showing the non-intersecting property between a forbidden area rising from the left segment and a forbidden area rising from the right segment. Their restricting pairs must be intertwined, and therefore we can use Monge property to derive the impossibility of their co-existence. This eventually implies the existence of a way of inserting the path $\pi_{a,b}$, and completes the greedy algorithm.

\section{Preliminaries}

By default, all graphs considered in this paper are \emph{directed} edge-weighted graphs.

%\paragraph{Okamura-Seymour instances.} 
Let $G$ be a plane graph (that is, a planar graph with its embedding). For a pair $u,v$ of its vertices, we denote by $\dist_G(u,v)$ the \emph{directed} distance in $G$ from $u$ to $v$ (that is, the minimum length of any directed path in $G$ from $u$ to $v$). We may omit the subscript $G$ when it is clear from the context.
It is clear that $\dist_G(\cdot,\cdot)$ is a quasi-metric on the vertices of $G$.

Let $T$ be a set of its vertices called \emph{terminals}. We say that $(G,T)$ is an Okamura-Seymour instance, iff in the planar drawing of $G$, all terminals lie on the boundary of the outer face.
%If $G$ is a directed graph, denoted by $\vG$, we say that $(\vG,T)$ is a directed Okamura-Seymour instance. 
We say that two pairs $(t_1,t_2),(t_3,t_4)$ of terminals \emph{cross}, iff $t_1,t_2,t_3,t_4$ are distinct terminals, and the order in which they appear on the boundary of the outer face is either $t_1,t_3,t_2,t_4$ or $t_1,t_4,t_2,t_3$.

\paragraph{Demands and routings.}
A \emph{demand} on a directed instance $(\vG,T)$ is a function $\D$ that assigns to each ordered pair $(t,t')$ of terminals an integer $\D(t,t')\ge 0$. 
Alternatively, we can view a demand $\D$ as a collection of pairs of terminals in $T$, where every pair $(t,t')$ appears $\D(t,t')$ times in the collection.
We say that $\D$ is \emph{unweighted} iff $\D(t,t')$ is either $0$ or $1$ for all $(t,t')$.
%We say that a \emph{flow} $F$ in a nest $(\vec G,T)$ \emph{realizes} the demand $D$, iff for every pair $t,t'$ of terminals, $F$ sends $D(t,t')$ units of flow from $t$ to $t'$. If $F$ does not exceeds the capacities in $G$, i.e., the total amount of flows
%In other words, an unweighted demand $D$ is simply a collection of pairs of terminals.
%
For an unweighted demand $\D$, a \emph{routing} of $\D$ in $G$ is a collection $\qset$ of directed paths in $\vG$ that contains, for each pair $(t,t')$ with $\D(t,t')=1$, a directed path in $\vG$ connecting $t$ to $t'$, such that all paths in $\qset$ are edge-disjoint.
We say that $\D$ is \emph{routable} in $G$ iff it has a routing in $G$.

\paragraph{Partial quasi-metrics.}
A \emph{partial quasi-metric} is a function $D: T\times T\to \mathbb{R}^+ \cup \set{*}$, such that
\begin{itemize}
	\item for all $t,t'\in T$, $D(t,t')= 0$ iff $t=t'$; and
	\item for all $t,t',t''\in T$, such that none of $D(t,t''),D(t,t'),D(t',t'')$ is $*$,
	\[D(t,t'')\le D(t,t')+D(t',t'').\]
\end{itemize}
Intuitively, partial quasi-metrics are quasi-metrics whose values are not yet completely determined, and for those values that are already determined, the triangle inequalities need to be satisfied.
For a graph $G$ that contains $T$, we say $G$ \emph{satisfies/realizes} a partial quasi-metric $D$, iff for every ordered pair $t,t'\in T$,
either $\dist_{\vG}(t,t')=D(t,t')$, or $D(t,t')=*$.

\section{Preparation: Nests, Edge-weight LP, and Restricting Pairs}

Let $D$ be a quasi-metric on $T$ and $\sigma$ a circular ordering on $T$ that satisfy the property in \Cref{thm: directed 4 point condition}. We will construct an Okamura-Seymour instance whose induced shortest path distance metric on $T$ is exactly $D$.
To be more specific, we will construct a specific type of Okamura-Seymour instances, called \emph{nests}, via a greedy approach, called \emph{path insertion}, that we introduce below.

\subsection{Graph structure: nests and Path Insertion}

\paragraph{Nests.} We say a directed Okamura-Seymour instance $(\vG,T)$ is a \emph{nest}, iff $\vG$ consists of a collection $\Pi$ of directed paths, such that 
\begin{itemize}
	\item every path $\pi\in \Pi$ has both endpoints in $T$;
	%\item the intersection between every pair of paths in $\pset$ is either empty or a single vertex; and
	\item the paths in $\Pi$ are edge-disjoint; and
	\item $E(G)=\bigcup_{\pi \in \Pi}E(\pi)$.
\end{itemize}
When the set $T$ of terminals are clear from the context, we also say that $\vG$ is a nest.

\paragraph{Path Insertion.}
We now describe a generic process of constructing a nest, called \emph{path insertion}. We start with an empty disc with all terminals in $T$ lying on its boundary according to the ordering $\sigma$. The path insertion process performs iterations. In each iteration, we insert a directed path named $\pi_{t,t'}$ connecting an ordered pair $t,t'$ of terminals in $T$, and draw it as a directed curve $\gamma$ within the disc connecting its endpoints, such that 
\begin{itemize}
\item $\gamma$ is a simple curve from $t$ to $t'$;
\item for every other curve $\gamma'$ inserted before, curves $\gamma,\gamma'$ intersect at a finite number (could be $0$, in which case they are disjoint) of points, called their \emph{crossings}; and
\item at each crossing, the curves indeed cross each other (but not just touch and then depart).
\end{itemize}

%For each crossing $p$, if we denote by $\tau_1,\tau_2$ ($\tau'_1,\tau'_2$, resp.) the segments of $\gamma$ ($\gamma'$, resp.) incident to $p$, then the segments appear around $p$ in the circular order $\tau_1,\tau'_1,\tau_2,\tau'_2$ (that is, curves $\gamma,\gamma'$ indeed cross but not just touch at each of their intersections).
%\znote{just say cross but not touch}

At any time over the course of a path insertion process, a nest is obtained by viewing each crossing as a vertex and viewing each curve segment between a pair of its consecutive crossings as a directed edge.
See \Cref{fig: nest} for an illustration.
%In this paper, we will use the terms crossings and intersections interchangeably.

\begin{figure}[h!]
	\centering
	\subfigure[The path insertion process inserts $3$ paths, shown in red, green, and purple, respectively.]
	{\scalebox{0.12}{\includegraphics{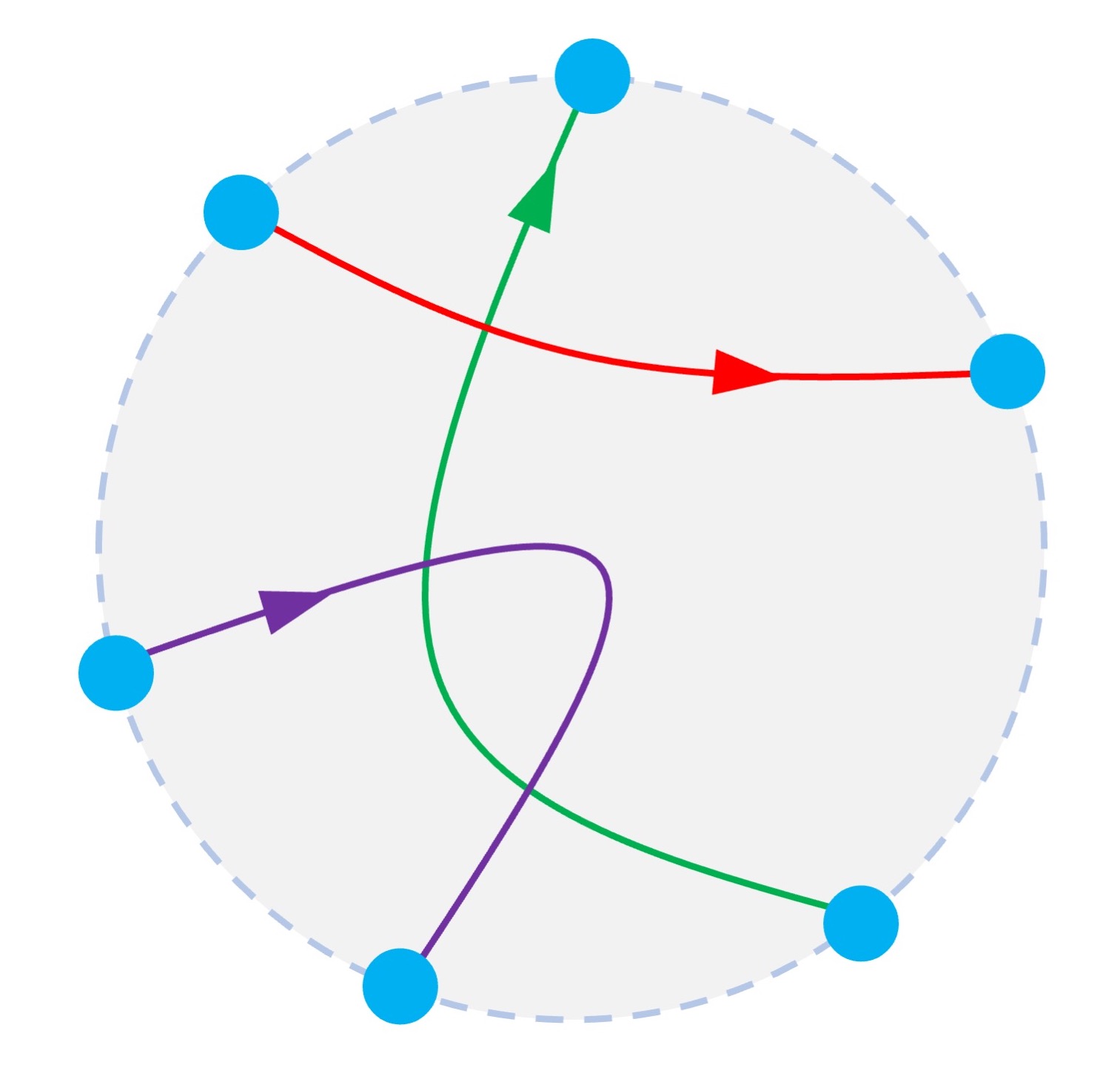}}}
	\hspace{1.7cm}
	\subfigure[The nest obtained from the path insertion process is a directed graph on $9$ vertices.]
	{
		\scalebox{0.12}{\includegraphics{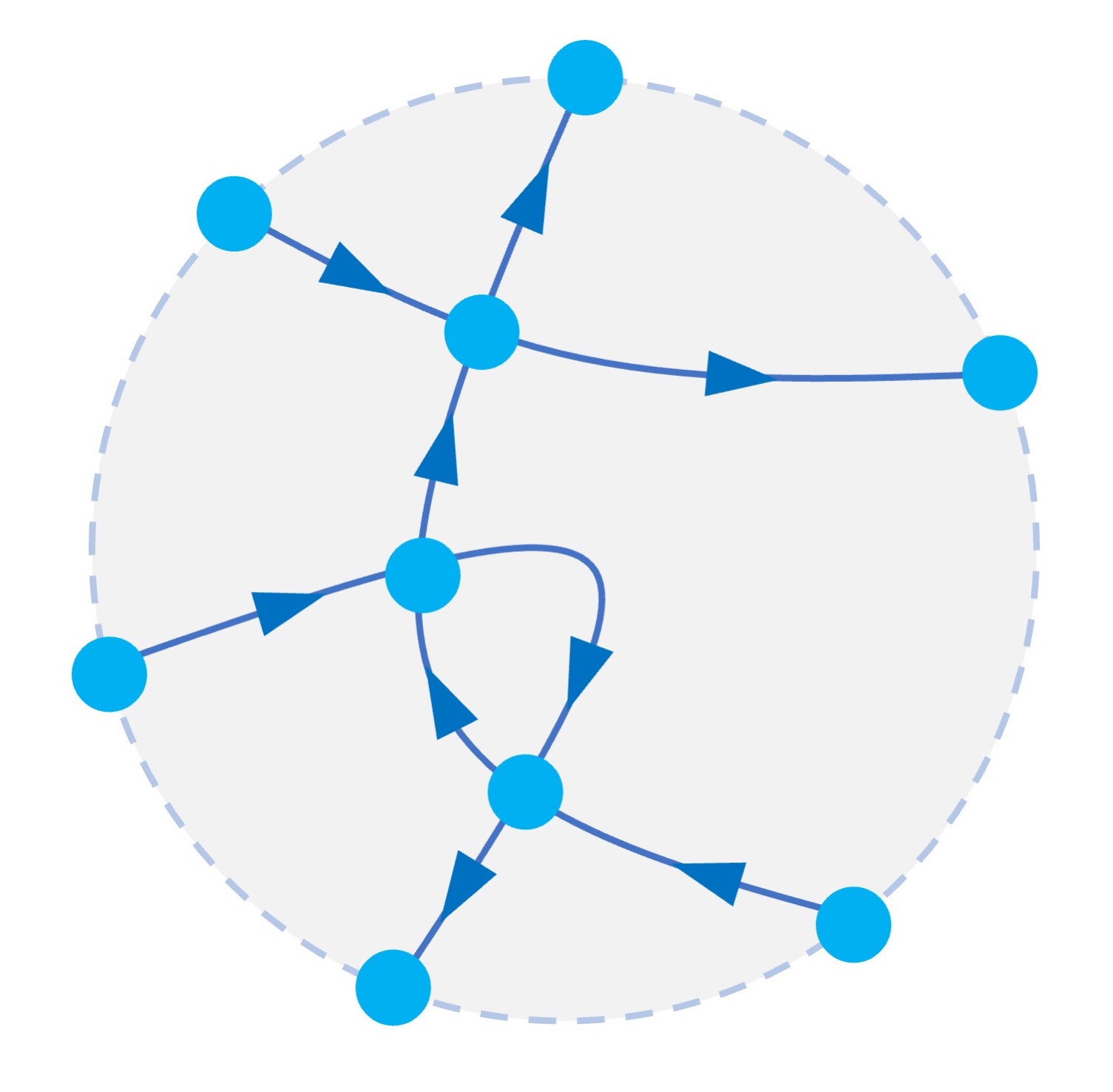}}}
	\caption{An illustration of a path insertion process and its produced nest.\label{fig: nest}}
\end{figure}

We view the input quasi-metric $D$ as given to us in a stream: we sequentially receive $|T|\cdot(|T|-1)$ tuples, each consisting of an ordered pair $(t,t')$ and their value $D(t,t')$. Equivalently, we are given a partial quasi-metric $D$ which initially assigns  $*$ to all pairs, and then it gets refined entry by entry until we get the complete quasi-metric $D$.

We will run the path insertion process along with the stream. Upon receiving the value $D(t,t')$ for a pair $(t,t')$, we insert the path $\pi_{t,t'}$ from $t$ to $t'$. Let $G$ be the nest produced by all paths inserted so far, so $G$ evolves with the stream. The goal is to ensure the following property:
\begin{properties}{P}
\item at any time in the stream, there exists a way of setting the edge weights of $G$ that makes it satisfy the current $D$; moreover, for every $(t,t')$, $\pi_{t,t'}$ is the shortest $t$-$t'$ directed path in $G$.
\label{prop}
\end{properties}
Note that, if this property is satisfied, then at the end of the process, we obtain a nest $G$, such that there exists a way of setting its edge weights to satisfy the whole quasi-metric $D$, and we are done.

\subsection{Setting the edge weights: LP and its dual}

For a nest $G$ to satisfy a partial quasi-metric $D$, the existence of edge weights can be characterized by the following LP (a feasibility LP without any objective function). We say that a pair $(t,t')$ is \emph{seen} if $D(t,t')\ne *$ (i.e., we have already received the pair $(t,t')$ and have inserted $\pi_{t,t'}$ in $G$).
\begin{eqnarray*}
	\mbox{(LP-Primal)}\quad	\quad 
	&\sum_{e\in E(\pi_{t,t'})}x_e\le D(t,t') &\forall \text{ seen }(t,t')\\
	&\sum_{e\in E(P)}x_e\geq D(t,t') &\forall \text{ seen }(t,t'),\forall \text{ di-path }P\text{ from }t \text{ to }t'\\
	&x_e\geq 0&\forall e\in E(G)
\end{eqnarray*}
By Farkas' lemma, the feasibility of (LP-Primal) is directly related to the feasibility of its dual.
To exploit it, we need the following flow-related definitions.
Let $\mathcal{P}$ be the collection of paths in $G$ with both endpoints in $T$. A \emph{terminal flow} $F: \mathcal{P}\to \mathbb{R}^+$ assigns each path $P\in\mathcal{P}$ with a value $F(P)$. 
For each edge $e\in E(G)$, we denote by $F_e$ the total amount of flow of $F$ sent through $e$. 
For a pair $F,F'$ of terminal flows, we say that $F$ \emph{dominates} $F'$ iff $F_e\ge F'_e$ for all $e\in E(G)$. 
For each pair $t,t'$ of terminals, we define $F_{t,t'}$ as the total amount of flow in $F$ from $t$ to $t'$. 
%We say that a pair $F,F'$ of flows are \emph{equivalent} if for each pair $t,t'\in T$, $F_{t,t'}=F'_{t,t'}$. 
We define the \emph{cost} of $F$ as $\cost(F)=\sum_{t,t'\in T}F_{t,t'}\cdot D(t,t')$.

The following claim is an immediate corollary of Farkas' lemma \cite{farkas1898fourier} and the fact that $\mathbb{R}$ is dense.
\begin{claim}
	\label{clm: feasibility or flow}
	%Let $Ax\le b$ be the canonical form of the linear program above, and suppose the latter assertion in Farkas' lemma holds. Then 
	Either \textnormal{(LP-Primal)} is feasible, or    
	there exist terminal flows $F,F'$ in $G$, such that:
	\begin{itemize}
		\item $F,F'$ assign integral values to all paths;
		\item $F$ only assigns non-zero values to paths in $\set{\pi_{t,t'}\mid (t,t') \text{ seen}}$;
		\item $F$ dominates $F'$; and
		\item $\cost(F)<\cost(F')$.
	\end{itemize}\label{cl:farkas}
\end{claim}

We will prove the following claim that, combined with
\Cref{clm: feasibility or flow}, immediately implies that Property \ref{prop} can be satisfied and completes the proof of \Cref{thm: directed 4 point condition}.

\begin{claim}
\label{clm: drawing}	
Upon receiving any pair $(t,t')$, there exists a way of inserting the path $\pi_{t,t'}$ to the current nest, such that in the resulting nest $G$, there are no flows $F,F'$ satisfying the properties of \Cref{cl:farkas}.
\end{claim}

%Let us look at the properties of flows $F,F'$ in \Cref{clm: feasibility or flow} in detail. Flow $F$ is quite restrictive, since $F$ only assigns non-zero values to paths in $\set{\pi_{t,t'}\mid (t,t') \text{ seen}}$. This means that, whenever the values $\set{F_{t,t'}\mid (t,t') \text{ seen}}$ are fixed, the whole flow $F$ is fixed as well: $F$ has to send $F_{t,t'}$ along path $\pi_{t,t'}$ for each seen pair $(t,t')$). Flow $F'$ is relatively flexible as it is allowed to use any path, but it needs to be dominated by $F$. So if we consider the graph where every edge $e$ has capacity $F(e)$, then $F'$ cannot exceed the capacity of any edge.

We say that flows $F,F'$ are \emph{aligned}, iff for each terminal $t$, the total amount of flow in $F$ sent from (received by, resp.) $t$ is identical to the total amount of flow in $F'$ sent from (received by, resp.) $t$.
We prove the following observation that allows us to focus on ruling out aligned flows $F,F'$.

\begin{observation}
	\label{obs: aligned}
	In order to prove \Cref{clm: drawing}, it suffices to ensure that, in the resulting nest $G$, there are no aligned flows $F,F'$ satisfying the properties of \Cref{cl:farkas}.
\end{observation}
\begin{proof}
For a pair of flows $F,F'$ satisfying the conditions in \Cref{cl:farkas}, we show that we can augment $F'$ to $F''$ such that flows $F,F''$ are aligned and still satisfy the conditions in \Cref{cl:farkas}.
Therefore, in order to prove \Cref{clm: drawing}, it is sufficient to ensure that there do not exist such aligned flows $F,F''$.

We construct a graph $H$ from $G$ as follows. We first give each edge $e\in E(G)$ capacity $F_e-F'_e$. We then add two new vertices $a,b$.
For each terminal $t$, we add a single arc from $a$ to $t$ with capacity $\sum_{t'\in T}F_{t,t'}-\sum_{t'\in T}F'_{t,t'}$ and a single arc from $t$ to $b$ with capacity $\sum_{t'\in T}F_{t',t}-\sum_{t'\in T}F'_{t',t}$. As $F$ dominates $F'$, $H$ is a graph with nonnegative integer weights.

We denote $\text{val}(F)=\sum_{t,t'\in T}F_{t,t'}$ and $\text{val}(F')=\sum_{t,t'\in T}F'_{t,t'}$.
First observe that there is an $a$-$b$ flow in $H$ with value $\text{val}(F)-\text{val}(F')$, and by our definition of edges capacities, such a flow saturates all edges incident to $a$ and $b$.
This is because, if we first send, for each $(t,t')$ and each $t$-$t'$ flow path $P$, $F(P)$ units of flow along the path $((a,t),P,(t',b))$, giving an $a$-$b$ flow of value $\text{val}(F)$, and then push back, along every flow path $P$ in its reverse direction, $F'(P)$ units of flow from $b$ to $a$, pushing back a total of $\text{val}(F')$ units, then eventually we obtain a flow of $\text{val}(F)-\text{val}(F')$ units from $a$ to $b$, that clearly satisfies all edge capacities in $H$.

We then compute the max-flow from $a$ to $b$ in $H$, which naturally induces a terminal flow $\hat F'$. We augment $F'$ with $\hat F'$ to obtain $F''$. Clearly, $F''$ and $F$ are aligned, and $F$ still dominates $F''$. As we have not modified $F$ at all, it still only sends flows along paths $\set{\pi_{t,t'}}$. Finally, as $F''$ is obtained from augmenting $F'$, it has a greater cost than $F'$, and so $\cost(F)<\cost(F')\le \cost(F'')$.
\end{proof}

Note that $F$ is completely determined by the values $\set{F_{t,t'}}$: $F$ has to send $F_{t,t'}$ units of flow along path $\pi_{t,t'}$ for each seen pair $(t,t')$. Comparatively, we can think of $F'$ as being determined in two steps: first we set the values $\set{F'_{t,t'}}$, then we view these values as a demand and find a routing of this demand in the $G$ where edges of each path $\pi_{t,t'}$ have capacity $F_{t,t'}$, as $F'$ needs to be dominated by $F$.
Note that the costs of $F$ and $F'$ are solely determined by the values $\set{F_{t,t'}}$ and $\set{F'_{t,t'}}$.

%We now define a central notion in preparation for the proof of \Cref{clm: drawing}. 
Assume that we have just received the terminal pair $(a,b)$ and value $D(a,b)$. As a preparation step for the proof of \Cref{clm: drawing}, we will first translate \Cref{clm: drawing} into the language of demands and routings, for which we need the following central notion.

\paragraph{Restricting Pairs.}
Let $\cset,\cset'$ be demands on $T$. We say that $\cset$ and $\cset'$ are \emph{aligned}, iff for each $t$,
\begin{itemize}
	\item the number of pairs $(t,\cdot)$ in $\cset$ equals the number of pairs $(t,\cdot)$ in $\cset'$; and 
	\item the number of pairs $(\cdot,t)$ in $\cset$ equals the number of pairs $(\cdot,t)$ in $\cset'$.
\end{itemize}

For a demand $\cset$, we define the edge-capacitated graph $G_{\cset}$ to be the nest $G$ such that edges of each path $\pi_{t,t'}$ have capacity $\cset(t,t')$, and we define 
$D(\cset)=\sum_{(t,t')\in \cset}D(t,t')$.

We say a pair of demands $(\cset,\cset')$ is an \emph{$(a,b)$-restricting pair}, iff
\begin{properties}{R}
	\item $\cset$ only contains seen pairs, and moreover, it contains the pair $(a,b)$;
	%both $a$ and $b$ appear only in pairs $(a,b)$, i.e., $\cset$ does not contain any pair $(a,t)$ where $t\ne b$ or any pair $(t,b)$ where $t\ne a$; and 
	%\item in $\cset'$, there is no pair $(a,b)$;
	\item demands $\cset,\cset'$ are aligned; and
	\item $D(\cset) < D(\cset')$.\label{prop: cost}
\end{properties}
%Note that all $(a,b)$-restricting pairs are  determined solely by the input quasi-metric $D$.

%For a collection $\cset$ of ordered terminal pairs, we define $\D(\cset)$ as the demand where for each pair $(t,t')$, $\D(\cset)(t,t')$ equals the number of times that the pair $(t,t')$ appear in $\cset$.

From \Cref{obs: aligned} and the above discussion, \Cref{clm: drawing} is equivalent to the following claim.

\begin{claim}
	\label{clm: demand}
	Upon receiving the pair $(a,b)$, there exists a way of inserting path $\pi_{a,b}$ such that in the resulting nest $G$, there are no $(a,b)$-restricting pair $(\cset,\cset')$ where $\cset'$ is routable in $G_\cset$.
\end{claim}

%\Cref{clm: demand} is the technical and the rest of the paper is dedicated to its proof. We provide a high-level overview in \Cref{sec: overview} and then its complete proof in \Cref{sec: Proof of clm: demand}.

\textbf{Remark.} If we allow the quasi-metric $D$ to take value $+\infty$. Then the proof is this section needs to be slightly adjusted to only inserting paths between pairs $(t,t')$ with $D(t,t')<+\infty$. Everything else stays the same.

\section{Completing the Proof of \Cref{thm: directed 4 point condition}}
\label{sec: Proof of clm: demand}

\newcommand{\bdl}{\partial_{\mathsf{L}}}
\newcommand{\bdr}{\partial_{\mathsf{R}}}

In this section, we provide the proof of \Cref{clm: demand}, thereby completing the proof of \Cref{thm: directed 4 point condition}. %

\subsection{Good sequences and good pairs}
\label{sec: good things}

First we introduce some definitions.
Recall that $a,b$ separate the boundary into two segments: the left segment $\bdl$ (from $a$ clockwise to $b$), and the right segment $\bdr$ (from $b$ clockwise to $a$). For points $x,x'$ both lying on the same segment, we say that $x$ \emph{precedes} $x'$ iff $x$ lies closer to $a$ than $x'$.

%Let $Q$ be a path connecting terminals $t,t'\in T$, where $(t,t')$ crosses $(a,b)$ (in this case we also say that $Q$ crosses $(a,b)$). So terminals $t,t'$ separates the boundary of the disc into two segments with only one containing $a$. We define $\area(Q)$ as the area enclosed by the union of (i) the $t$-$t'$ boundary segment that contains $a$; and (ii) path $Q$.

We say that a path $Q$ is \emph{crossing} iff its starting and ending points lie on different segments.
Let $Q$ ($Q'$, resp.) be a crossing path from $x$ to $y$ (from $x'$ to $y'$, resp.). We say that path $Q$ \emph{precedes} path $Q'$, iff either
\begin{itemize}
\item $x,x'$ lie on the same segment (and so do $y,y'$), $x$ precedes $x'$, and $y$ precedes $y'$; or 
\item $x,y'$ lie on the same segment (and so do $y,x'$), and $x$ precedes $y'$.
\end{itemize}

We say that a sequence $\qset=(Q_1,\ldots,Q_r)$ of crossing paths is \emph{good}, iff for each $i\ge 1$, $Q_i$ precedes $Q_{i+1}$.
Let $\qset$ be a good sequence and let $\dset'$ be a demand.
We say that $(\qset,\dset')$ is a \emph{good pair}, iff 
\begin{itemize}
	\item $\qset=(Q_1,\ldots,Q_r)$, where $Q_i$ connects $x_i$ to $y_i$; and
	\item $\dset'=\set{(a,y_1),(x_1,y_2),\ldots,(x_r,b)}$.
	%\item in the nest $G$, if we denote by $Q_i$ the $x_i$-$y_i$ path, then the sequence $(Q_1,\ldots,Q_r)$ is monotone.
\end{itemize}

Let $\qset$ be a set of paths and $\cset$ a demand on $T$. We say $\qset\preceq \cset$ iff  $\qset$ does not violate the capacity of $G_{\cset}$. That is, for each $e\in E(G)$,
$\big|\set{Q\in \qset\mid e\in E(Q)}\big|\le \sum_{(t,t')\in \cset}\cset(t,t')\cdot \mathbf{1}[e\in \pi_{t,t'}]$.

%\Cref{clm: demand} concerns routability of a demand $\cset'$ in the graph $G_{\cset}$.

\iffalse
We show in the next observation that, intuitively, good pairs enforce routability.

\begin{observation}
	\label{obs: good always routable}
	For any good pair $(\qset,\dset')$, no matter how the path $\pi_{a,b}$ is inserted to $G_{\qset}$, the nest formed by inserting all paths of $\qset$, demand $\dset'$ is always routable in the resulting nest.
\end{observation}
\begin{proof}
	Since the sequence $\qset$ is monotone, no matter how we insert the path $\pi_{a,b}$, for each $1\le i\le r$, if we denote by $x_i$ the first intersection between $\pi_{a,b}$ and $Q_i$, then the crossings $p_1,p_2,\ldots,p_r$ appear on $\pi_{a,b}$ in this order from $a$ to $b$.
	%We set these crossings as \emph{turn} crossings, and all others as \emph{pass} crossings. This will give us a routing of demand $\dset'$. To see why this is true,
	We now describe a routing of $\dset'$. 
	First the $a$-$y_1$ path: start it at $a$, we let it follow the trajectory of the path $\pi_{a,b}$ until it meets $p_1$, and then let it switch course to follow the trajectory $Q_1$ until it ends at $y_1$, so we get an $a$-$y_1$ path. 
	Similarly, for each $1\le i\le r-1$, the $x_i$-$y_{i+1}$ path starts from $x_i$, follows the trajectory of $Q_i$ until $p_{i+1}$, and then switches course to follow the trajectory of $Q_{i+1}$ until it reaches $y_{i+1}$. The $x_r$-$b$ path is constructed similarly.
\end{proof}
\fi

To prove \Cref{clm: demand}, we will show that the only restricting pairs $(\cset,\cset')$ that enforce the routability of $\cset'$ in $G_\cset$ are, in some sense, good pairs.
In other words, if \Cref{clm: demand} is not true, in that no matter how we insert path $\pi_{a,b}$, there is always some routable restricting pair, then this must be because, even before the insertion of $\pi_{a,b}$, we are already doomed by some restricting pair that is a combination of good pairs.
The formal statement is provided below in \Cref{lem: inevitable}. Here we let $\cset^-$ be the demand obtained from $\cset$ by removing all pairs $(a,b)$. For a path $Q$ from $t$ to $t'$, we define $D(Q)=D(t,t')$, and for a set $\qset$ of paths, we define $D(\qset)=\sum_{Q\in \qset}D(Q)$.

\begin{lemma}
	\label{lem: inevitable}
	Assume \Cref{clm: demand} is not true (that is, there exist a quasi-metric $D$ and a circular ordering $\sigma$ on $T$, and an implementation of path insertion process, such that upon receiving  $(a,b)$, no matter how we insert the path $\pi_{a,b}$, in the resulting nest $G$, there is some $(a,b)$-restricting pair $(\bar \cset,\bar \cset')$ with $\bar \cset'$ routable in $G_{\bar \cset}$).
	Then in the current nest there exists an $(a,b)$-restricting pair $(\cset,\cset')$ and a set $\qset$ of paths, such that
	\begin{itemize}
		\item $\qset\preceq \cset^-$;
		\item $\qset=\qset_0\cup\big(\bigcup_{1\le i\le s}\qset_i\big)$, where for each $i\ge 1$, $\qset_i$ is a good sequence;
		\item $\cset'=\dset'_0\cup \big(\bigcup_{1\le i\le s}\dset'_i\big)$; 
		\item $\qset_0$ is a routing of $\dset'_0$; and
		\item for all $1\le i\le s$, $(\qset_i,\dset'_i)$ is a good pair.
	\end{itemize} 
\end{lemma}

\Cref{lem: inevitable} is the main technical lemma of the paper, whose proof is provided later. Now we complete the proof of \Cref{clm: demand} using \Cref{lem: inevitable}.

\subsubsection*{Proof of \Cref{clm: demand} using \Cref{lem: inevitable}}
From \Cref{lem: inevitable}, if \Cref{clm: demand} is not true, then in the current nest there exists an $(a,b)$-restricting pair and its decompositions $\qset=\bigcup_{0\le i\le s}\qset_i$, $\cset'=\bigcup_{0\le i\le s}\dset'_i$, such that
for each index $1\le i\le s$, $(\qset_i,\dset'_i)$ is a good pair.
We now show that Monge property does not allow such a restricting pair. Specifically, we will show that if the quasi-metric $D$ satisfies the Monge property, then $D(\cset)\ge D(\cset')$, causing a contradiction to Property~\ref{prop: cost}.

First, since $\cset^-$ belongs to the nest before we insert path $\pi_{a,b}$, $\qset\preceq \cset^-$ implies that $D(\qset)\le D(\cset^-)$.
(Here we assume that the current nest is the first time (in the Path Insertion process) where \Cref{clm: demand} is not true.)
Since $\qset_0$ is a routing of $\dset'_0$, $D(\qset_0)=D(\dset'_0)$. For each $1\le i\le s$, let $\dset_i$ be the union of (i) the demand induced by $\qset_i$; and (ii) an additional pair $(a,b)$, so $\dset_i$ is aligned with $\dset'_i$. We now show that Monge property ensures that $D(\dset_i)\ge D(\dset'_i)$, which this implies that $D(\cset)\ge D(\cset')$, as  
\[
\begin{split}
D(\cset) & = D(\cset^-)+s\cdot D(a,b)\ge D(\qset)+s\cdot D(a,b)=\sum_{0\le i\le s}D(\qset_i)+s\cdot D(a,b)\\& = D(\qset_0)+\sum_{1\le i\le s}\bigg(D(\qset_i)+ D(a,b)\bigg)
= D(\dset'_0)+\sum_{1\le i\le s}D(\dset_i) \ge D(\dset'_0)+\sum_{1\le i\le s}D(\dset'_i) \ge D(\cset').
\end{split}
\]

Note that $\dset_i=\set{(a,b),(x_1,y_1),\ldots,(x_r,y_r)}$ and $\dset'_i=\set{(a,y_1),(x_1,y_2),\ldots,(x_r,b)}$. 
Since $Q_1$ is a crossing path, the pair $(x_1,y_1)$ crosses $(a,b)$, so from Monge property, 
$$D(a,b)+D(x_1,y_1)\ge {\color{red}D(x_1,b)}+D(a_1,y_1).$$
Then, since $Q_1$ precedes $Q_2$, either 
\begin{itemize}
\item $x_1,x_2$ lie on the same segment, and $x_1$ precedes $x_2$, which means $(x_1,b)$ crosses $(x_2,y_2)$; or 
\item $x_1,y_2$ lie on the same segment, and $x_1$ precedes $y_2$, which also means $(x_1,b)$ crosses $(x_2,y_2)$.
\end{itemize} 
Therefore, from Monge property, 
$${\color{red}D(x_1,b)}+D(x_2,y_2)\ge {\color{blue}D(x_2,b)}+ D(x_1,y_2).$$
Similarly, we can deduce that $(x_2,b)$ crosses $(x_3,y_3)$, so from Monge property, 
$${\color{blue}D(x_2,b)}+D(x_3,y_3)\ge {\color{green}D(x_3,b)}+ D(x_2,y_3).$$
We apply similar arguments to each $i$ and then sum up all inequalities, getting that
$$D(\dset_i)=D\big((a,b),(x_1,y_1),(x_2,y_2),\ldots,(x_r,y_r)\big)\ge D\big((a,y_1),(x_1,y_2),(x_2,y_3),\ldots,(x_r,b)\big)=D(\dset'_i).$$

\subsection{Proof of \Cref{lem: inevitable}}

Let $G$ be the nest before the iteration of inserting the path $\pi_{a,b}$.
The \emph{trajectory} of $\pi_{a,b}$ is defined to be the sequence of faces of $G$ that it intersects. 
Here we carve out a tiny disc around $a$, making it face $F_a$, and define face $F_b$ for $b$ similarly, so an $a$-$b$ curve/path always starts at $F_a$ and ends at $F_b$.
As $\pi_{a,b}$ is not allowed to contain any intersections in $G$, the trajectory characterizes its drawing. Therefore, we do not distinguish between a path/curve and its trajectory. 

Throughout the proof, we will not only consider curves connecting $a$ to $b$ but also curves starting at $a$ and ending at some face inside the disc.
We will prove \Cref{lem: inevitable} by induction on these $a$-starting curves. Due to technical reasons, in addition to trajectories, we will also describe these $a$-starting curves by \emph{walls}, which intuitively form corridors that govern the trajectories of curves. The notions of good sequences and good pairs also need to be adjusted with walls.

\paragraph{Walls, and growing the walls.}
Let $(F_a,F_1,\ldots,F_{r},F_b)$ be the trajectory of an $a$-$b$ path. 
Denote $F_0=F_a$ and $F_{r+1}=F_b$ for convenience.
For each $i\ge 1$, let $e_i$ be the edge shared by faces $F_{i-1}$ and $F_i$, so edges $e_i, e_{i+1}$ separate the boundary of face $F_i$ into two segments: one from $e_i$ clockwise to $e_{i+1}$ (excluding both), that we call the \emph{left wall} of $F_i$, and the other from $e_{i}$ counter-clockwise to $e_{i+1}$ (excluding both), that we call the \emph{right wall} of $F_i$.
We define left/right walls for $F_0,F_{r+1}$ similarly.
The left wall of the trajectory is defined to be the union of left walls of all faces $F_0,\ldots,F_{r+1}$, and the right wall of the trajectory is defined symmetrically.
We call such walls \emph{complete walls}, and we call the edges $e_1,\ldots, e_r$ \emph{corridor edges}.
 See \Cref{fig: wall}.

Alternatively, we can also think of the wall as being built gradually, as we go from $F_a$ to $F_b$ along the curve.
We initiate the left/right walls as $\wl,\wri=\emptyset$, and we place a \emph{guiding point} on each of them, which is initially at the left/right intersection between the tiny $a$-disc and the boundary (i.e., the endpoints of ``edge $e_0$''). As we enter the face $F_i$, the left guiding point starts to move from the endpoint of $e_i$ to along the left wall of $F_i$ to the endpoint of $e_{i+1}$, and the right guiding point move similarly on the right wall.
After they both reach the endpoints of $e_{i+1}$, we move from face $F_i$ to face $F_{i+1}$.
We keep moving guiding points until they both reach the intersections between the tiny $b$-disc and the boundary (i.e., the endpoints of ``edge $e_{r+1}$''). The wall-growing process can be naturally broken down into discrete steps: at each step, we either let the left guiding point or let the right guiding point travel along one edge.

As we will consider general $a$-starting curves (not necessarily  ending at $b$), we think of the above wall-growing process as terminable at any step. At each step, the left (right, resp.) wall is defined to be the set of all edges that the left (right, resp.) guiding point has travelled so far. For convenience, we will also refer to this ``incomplete walls'' as walls.

\begin{figure}[h!]
	\centering
	\subfigure[A complete wall from $a$ to $b$, with some corridor edges shown in orange.]
	{\scalebox{0.14}{\includegraphics{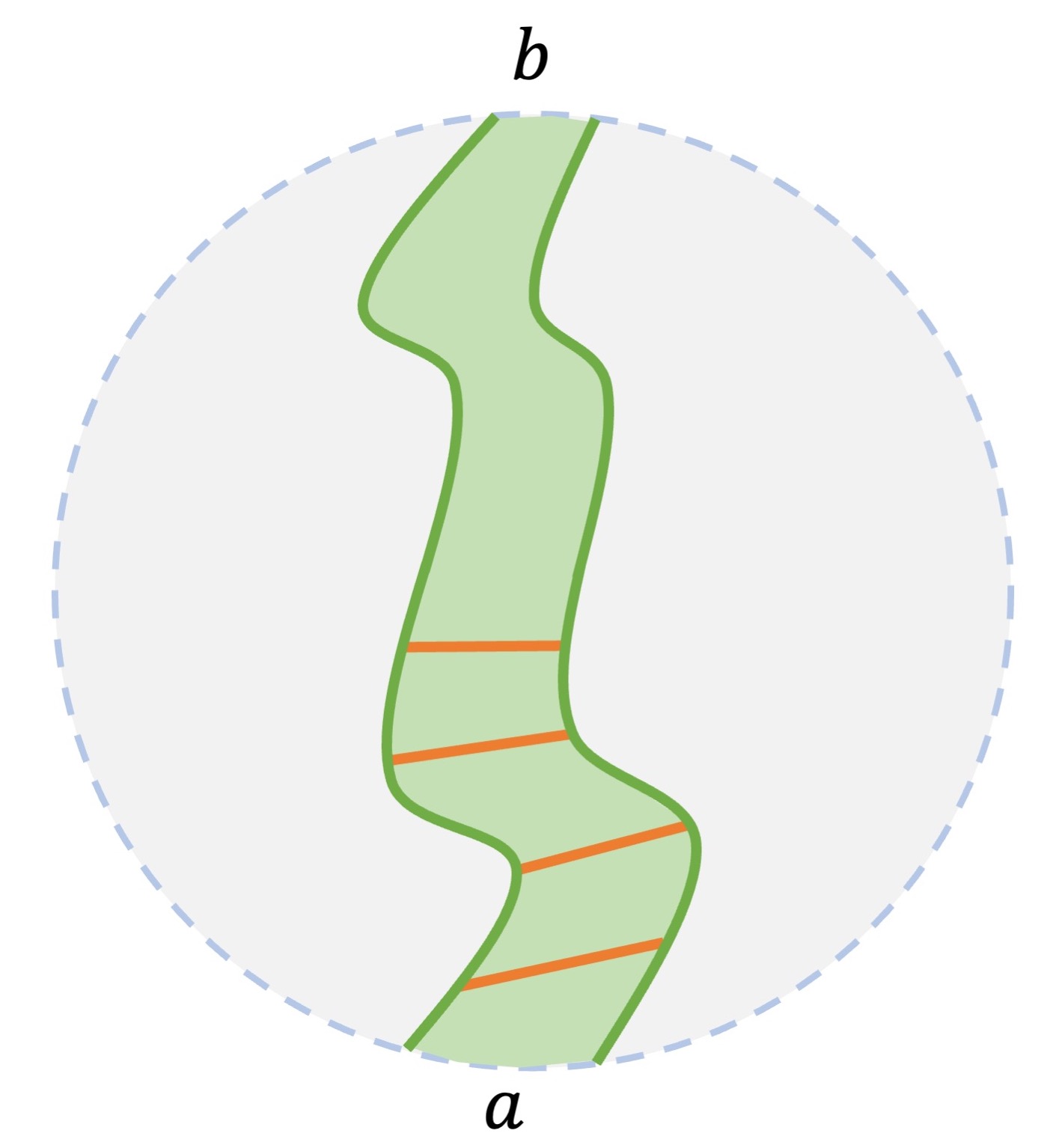}}}
	\hspace{1.4cm}
	\subfigure[An incomplete wall. The paths $Q,Q',Q''$ have their endpoints $x,x',x''$ on the right wall, where $x$ precedes $x''$, and $x''$ precedes $x'$. The Euler tour is shown in purple dashed line.]
	{
		\scalebox{0.14}{\includegraphics{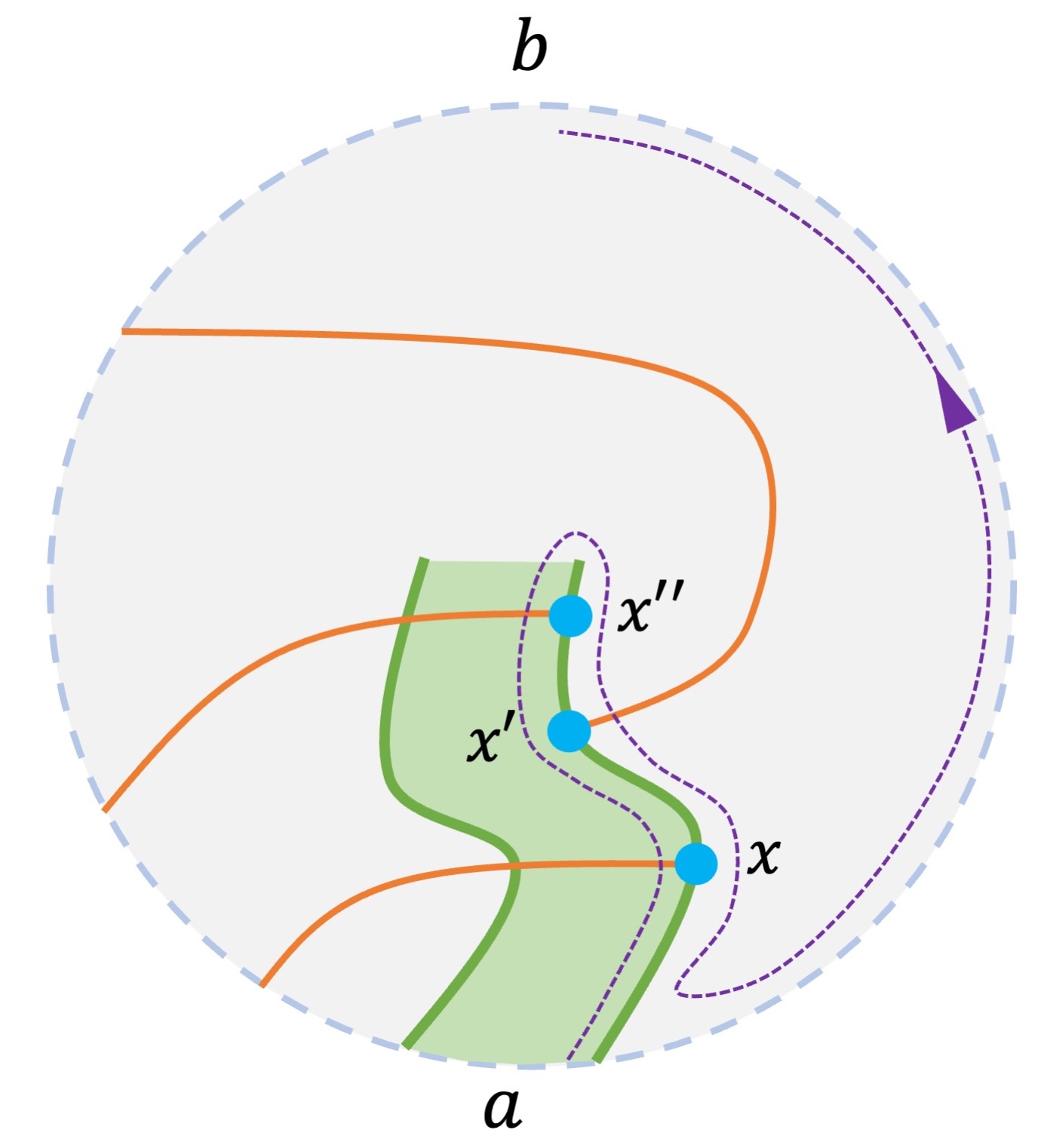}}}
	\caption{An illustration of complete and incomplete walls.\label{fig: wall}}
\end{figure}

\paragraph{Good sequences, good pairs and good tuples with respect to walls.}
%Recall that the boundary is separated by points $a,b$ into two segments: the left segment $\bdl$ (from $a$ clockwise to $b$) and the right segment $\bdr$ (from $a$ counter-clockwise to $b$).
Let $W=(\wl,\wri)$ be a wall.
We say a path $Q$ is crossing iff it has one endpoint in $\bdl\cup \wl$ and the other endpoint in $\bdl\cup \wri$. %Among the regions of the disc separated by $Q\cup (\bdl\cup \wl) \cup (\bdl\cup \wri)$, we define $\area_{W}(Q)$ as the one that contains $a$. Clearly, when $\wl=\wri=\emptyset$, $\area_{W}(Q)=\area(Q)$.
We now call $\bdl\cup \wl$ the \emph{left segment} and $\bdr\cup \wri$ the \emph{right segment}.
Let $Q$ ($Q'$, resp.) be a $W$-crossing path from $x$ to $y$ (from $x'$ to $y'$, resp.). We say that $Q$ \emph{$W$-precedes} $Q'$, iff
either
\begin{itemize}
	\item $x,x'$ lie on the same segment (say $\bdl\cup \wl$, and so $y,y'$ lie on $\bdr\cup \wri$), and $x$ \emph{$W$-precedes} $x'$, which is defined as follows: as we move along the Euler-tour around the left segment, starting from the ``corridor side'' of $\wl$ and ending at $b$, we encounter the $x$-incident edge in $Q$ before the $x'$-incident edge of $Q'$ (see \Cref{fig: wall}); and additionally, $y$ $W$-precedes $y'$; or 
	\item $x,y'$ lie on the same segment (and so do $y,x'$), and $x$ precedes $y'$.
\end{itemize}

We say that a sequence $\qset=(Q_1,\ldots,Q_r)$ of crossing paths is \emph{$W$-good}, iff for each $i\ge 1$, path $Q_i$ $W$-precedes path $Q_{i+1}$.
Let $\qset$ be a $W$-good sequence and let $\dset'$ be a demand.
%
%We say the boundary that is between $a$ and $b$ counter-clock wise as the left boundary, the boundary between $a$ and $b$ clock wise as the right boundary. For any two point $c$ and $d$, we say $c$ above $d$ if they are on the same side of the boundary and $(a,c,d,b)$ is the order on the boundary, and in this case we also say $d$ is below $c$.
%
We say that $(\qset,\dset')$ is a \emph{$W$-good pair}, iff
\begin{itemize}
	\item for each $1\le i\le r$, path $Q_i$ is from $x_i$ to $y_i$; and
	\item $\dset'=\set{(a,y_1),(x_1,y_2),\ldots,(x_r,b)}$.
\end{itemize}
%Clearly, when $\gamma$ is the empty $a$-$a$ curve, $\gamma$-good pairs are simply good pairs defined before.
It is easy to verify that, when the wall $W$ is empty (that is, $\wl=\wri=\emptyset$), all $W$-related notions degenerate to the notions (good sequences, good pairs) defined at the beginning of \Cref{sec: good things}.

\paragraph{Switching.} The major operation we will perform in constructing $W$-good sequences is \emph{switching}. Let $Q,Q'$ be paths intersecting at $p$. Denote by $Q_1$ the subpath of $Q$ between its start point and $p$, and by $Q_2$ the subpath of $Q$ between $p$ and its end point. We define $Q'_1,Q'_2$ similarly. We say that the concatenated paths $Q_1\cup Q'_2$ and $Q_2\cup Q'_1$ are obtained by \emph{switching paths $Q_1,Q_2$ at the intersection $p$.} For a set $\qset$ of paths, switching $\qset$ means a series of switching operations on pairs of paths in $\qset$.
We prove the following claim, which will be useful for our proof.

\begin{claim} \label{clm:self-config}
	Let $\qset$ be a $W$-good sequence, and
	let $\qset'$ be obtained from $\qset$ by replacing some path $Q_i$ from $x_i$ to $y_i$ with another path $Q'_i$ from $x_i$ to $y'_i$, where $y_i$ precedes $y'_i$. 
	Then we can switch $\qset'$ to obtain another $W$-good pair.
\end{claim}

\begin{proof}
	Denote $\qset=(Q_1,\ldots,Q_r)$, where $Q_i$ connects $x_i$ to $y_i$ for each $i\ge 1$, and path $Q'_i$ connects $x_i$ to $y'_i$.
	Let $\dset=\set{(a,y_1),(x_1,y_2),\ldots,(x_r,b)}$ be the associated demand of $\qset$. 
	Since $y_i$ precedes $y'_i$, path $Q_{i-1}$ precedes $Q'_{i}$. If path $Q'_i$ precedes $Q_{i+1}$, then $\qset'$ is still a $W'$-good sequence and we are done.
	If $Q'_{i}$ no longer precedes $Q_{i+1}$, then $x_i$ and $y_{i+1}$ do not lie on the same segment, so $x_i, x_{i+1}$ lie on the same segment, and therefore $x_i$ precedes $x_{i+1}$, $y_i$ precedes $y_{i+1}$, and $y_{i+1}$ precedes $y'_i$.
	Therefore, $(x_i,y'_i)$ and $(x_{i+1},y_{i+1})$ cross. We switch the intersection between paths $Q'_i$ and $Q_{i+1}$, obtaining a path from $x_i$ to $y_{i+1}$ satisfying the demand $(x_i,y_{i+1})$ in $\dset$ and another path from $x_{i+1}$ to $y'_i$, which we denote by $Q'_{i+1}$, and replace paths $Q'_i,Q_{i+1}$ in $\qset$ with $Q'_{i+1}$. Comparatively, $Q'_{i+1}$ can be seen as obtained from $Q_{i+1}$ by moving its endpoint from $y_{i+1}$ down to $y'_i$. Therefore, we can repeat this argument for paths $Q'_{i+1}$ and $Q_{i+2}$ and so on, until either some $Q'_{j}$ precedes $Q_{j+1}$ or we switch all the way to get $Q'_r$ as the last path. Either case, we obtain a $W$-good sequence.
\end{proof}

%See {\color{red}Figure}.

%Let $G^\gamma_{\cset}$ be the nest obtained from $G_{\cset}$ by inserting the curve $\gamma$, where each edge of $\gamma$ has capacity $\cset(a,b)$.

We say that a tuple $(\qset,\cset,\cset')$ is a \emph{$W$-good tuple}, iff
\begin{itemize}
\item $(\cset,\cset')$ is an $(a,b)$-restricting pair;
\item $\cset'$ can be wall-respecting split into $\dset$ (that is, $\dset$ is obtained from $\cset'$ by replacing each demand $(x,y)$ with a set $\set{(x,z_1),(z_1,z_2),\ldots,(z_t,y)}$ of demands, where each $z_i$ lies on either the left segment or the right segment); and $\dset$ can be decomposed into $\dset=\dset_0\cup\big(\bigcup_{1\le i\le s}\dset_i\big)$;
\item $\qset\preceq G^{\gamma}_{\cset}$, and a decomposition $\qset=\qset_0\cup\big(\bigcup_{1\le i\le s}\qset_i\big)$, such that
\begin{itemize}
\item $\qset_0$ is a routing of $\dset_0$; and
\item for each $1\le i\le s$, $(\qset_i,\dset_i)$ is a $W$-good pair.
\end{itemize}  
\end{itemize}

\paragraph{The structure of induction.}
We will prove \Cref{lem: inevitable} by induction, and show that, if \Cref{clm: demand} is not true, then in the current nest, for any wall $W$, there exists a $W$-good tuple. Clearly, \Cref{lem: inevitable} is the special case where $W$ is the empty wall (i.e., $\wl=\wri=\emptyset$).
It is also easy to see that the number of all possible walls is finite.
We will start from the cases where $W$ is a complete wall. For the inductive step, we consider any incomplete wall $W$. By definition of the wall-crowing process, we can only grow $W$ by either extending its left wall by an edge (and obtain $W'$) or extending its left wall by an edge (and obtain $W''$). Assuming both $W'$ and $W''$ admit good tuples, we will show that there is also a $W$-good tuple.

\paragraph{Base case.}
The base case is when $W$ is a complete wall. By definition, $W$ corresponds to an $a$-$b$ curve $\gamma$ whose trajectory is guided by $W$.
From the assumption (that Claim $6$ is false), if we insert $\gamma$ as the path $\pi_{a,b}$, then there is an $(a,b)$-restricting pair $(\bar{\cset},\bar{\cset}')$ with $\bar{\cset}'$ routable in $G_{\bar{\cset}}$. 

%In this state, all the entry edge along the ways are seperateing $a$ and $b$. For each $(a,b)$ curve in $\bar{C}$, the entry edges it switched with forms a good sequrence, and together with the segment formed after switch, it is a good pair. Thus It is a good tuple.
%
Consider the routing of $\bar{\cset}'$ in $G_{\bar{\cset}}$ and only focus on the demands that use part of $\pi_{a,b}$.
For each such demand $(z,w)$, assume its routing $Q=Q[z,p_1]\cup Q[p_1,p_2] \cup Q[p_2,w]$, where $Q[p_1,p_2]=Q\cap \pi_{a,b}$. Denote by $x_1,y_1$ the endpoints of the corridor edge that contains $p_1$, such that $x_1\in Q[z,p_1]$, and denote by $x_2,y_2$ the endpoints of the corridor edge that contains $p_2$, such that $y_2\in Q[p_2,w]$.
The demand $(z,w)$ can be then wall-respecting split into $(z,x_1),(x_1,y_2),(y_2,w)$, where $(z,x_1)$ is satisfied by $Q[z,x_1]$ and $(y_2,w)$ is satisfied by $Q[y_2,w]$. We are only left with a corridor edges $(x_1,y_1),(x_2,y_2)$, which are then added to $\qset'$. Since corridor edges have endpoints on left and right walls and do not cross each other, they automatically form $W$-good sequences.

%Inside the corridor, there is only ``switching with $\pi_{a,b}$'' operation, so the corridor edges that are switched in routing form a good set and then a good pair

\subsection*{Inductive step}

%As the walls are set up as the active points travel, we will have two sequence of walls. We call the walls set up by the active point starting on the left as the left wall, and the walls set up by the active point starting on the right as the right wall. Each step we can either prolong the left wall or right wall.
As the induction step, we will prove the following claim.

\begin{claim}
\label{clm: induction}
Let $W$ be a wall. Let $W'$ ($W''$, resp.) be the wall obtained from $W$ by growing an edge on its left (right, resp.) wall (i.e., allowing the left guiding point to travel one more edge).
If there is a $W'$-good tuple and a $W''$-good tuple, then there is a $W$-good tuple.
\end{claim}

From the definition of walls and their growing process, in each step we can either an edge to its left wall or its right wall, until the end it becomes a complete wall. Therefore, \Cref{clm: induction} together with the base case implies that for all walls $W$, there is a $W$-good tuple, which completes the proof of \Cref{lem: inevitable}.
The remainder of this section is dedicated to the proof of \Cref{clm: induction}.

Consider first the wall $W'$ obtained from $W$ by growing an edge on its left wall. Let $c$ be the new endpoint.
Let $(\qset',\cset',\bar\cset')$ be the $W'$-good tuple guaranteed by \Cref{clm: induction}, and we additionally assume that it is the one with the minimum number of paths in $\qset'$.
Denote $\qset'=\rset'_0\cup\big(\bigcup_{1\le i \le s}\rset'_i\big)$, where each $\rset'_i$ is a $W'$-good sequence.
\paragraph{Active sequences.}
We say that a sequence $\rset'_i$ is \emph{active} iff $c$ is an endpoint of some of its path. Note that, if $\rset'_i$ is inactive, then the corresponding $W'$-good pair $(\rset'_i,\dset_i)$ is also a $W$-good pair. 
Therefore, we only need to focus on active sequences.
Fix now an active sequence $\rset$. 
Let $R_i$ be the first path in $\rset$ with an endpoint $c$.
Recall that each path in $\rset$ is obtained from wall-respecting splitting the routing of some demand.
Let $R^*_i$ be the path obtained by the concatenation of (i) path $R_i$; and (ii) the $c$-incident subpath next to $R_i$ of the routing path containing $R_i$ (so $R^*_i$ contains $c$ as the only internal $W'$-vertex and no internal $W$-vertices).
 %If $c$ is not an endpoint of $R_{i+1}$, then $R^*_{i+1}=R_{i+1}$.

\paragraph{Case 1: $R^*_i$ is crossing.}
In this case, we can convert the $W'$-good sequence $\rset$ into a $W$-good sequence by simply replacing $R_i$ with $R^*_i$, with potential switching. 

\paragraph{Case 1.1: $R_i$ from left to right.}
In this case, $R_i$ is a suffix of $R^*_i$.
If $R^*_i\setminus R_i$ intersects $R_{i+1}$, then we switch paths $R^*_i$ and $R_{i+1}$, getting a path satisfying the demand $(x_i,y_{i+1})$ in $\dset$ and another path still containing $R_i$, obtaining a sequence with fewer paths. 
If $R^*_i\setminus R_i$ does not intersect $R_{i+1}$, then $R^*_i$ precedes $R_{i+1}$, so replacing $R_i$ with $R^*_i$ yields a $W$-good sequence.

\paragraph{Case 1.2: $R_i$ from right to left.}
In this case, $R_i$ is a prefix of $R^*_i$, so replacing path $R_i$ with $R^*_i$ is essentially moving down the endpoint of $R_i$ on the left segment. From \Cref{clm:self-config}, we can switch to obtain a $W$-good sequence.

\paragraph{Case 2: $R^*_i$ is non-crossing, and $R_i$ is a suffix of $R^*_i$.}
In this case, $R^*_i$ both starts and ends on the right segment. We will reduce it to the case where $R^*_i$ is non-crossing, and $R_i$ is a prefix of $R^*_i$.
Denote by $y_i$ the other endpoint of $R_i$, so $R_i$ connects $c$ to $y_i$. 

\paragraph{Case 2.1: $R^*\setminus R_i$ belongs to some $W'$-good sequence.} Denote the sequence by $\rset'=(R'_1,\ldots)$, where for each $t$, $R'_t$ connects $x'_t$ to $y'_t$, and $y'_j=c$, namely $R'_j=R^*\setminus R_i$.
So in the associated $\dset'$ for $\rset'$, there is a demand $(x'_{j-1},c)$. Since in the associated $\dset$ for $\rset$, there is a demand $(c,y_{i+1})$.
So the demands $(x'_{j-1},c)$ and $(c,y_{i+1})$ can be combined as one, i.e., $(x'_{j-1},y_{i+1})$. 
We reorganize the sequences $\rset,\rset'$ into two new sequences as follows.
\begin{itemize}
\item the first sequence is $(R_1,\dots,R_{i-1}, R^*_i, R'_{j+1},\ldots)$; and 
\item the second sequence is $(R'_1,\dots,R'_{j-1},R_{i+1},\ldots)$.
\end{itemize}
It is easy to verify that the associated demands for these two sequences correspond to the union of $\dset$ and $\dset'$ with $(x'_{j-1},c)$ and $(c,y_{i+1})$ combined into $(x'_{j-1},y_{i+1})$. 
So if these two sequence are both $W$-good, then we essentially decrease the number of paths in all good sequences.
 
For the first sequence:
\begin{itemize}
\item If $x'_j$ precedes $y_i$, we undo the combination of demands, leaving $R_i$ (instead of $R^*_i$) in the first sequence and making $R^*_i\setminus R_i$ satisfy the corresponding demand $(x'_j,c)$. Clearly, the sequence is $W$-good up to $R_i$. Consider now the path $R'_{j+1}$ which connects $x'_{j+1}$ to $y'_{j+1}$. 
\begin{itemize}
\item If $x'_{j+1}$ lies on the right segment, then $y'_{j+1}$ is on the left segment and is preceded by $c$, which means $R_i$ is above $(x'_{j+1},y'_{j+1})$ and the whole sequence is good.
\item If $x'_{j+1}$ lies on the left segment, then $y'_{j+1}$ is on the right segment and is preceded by $x'_j$.
Note that $c$ precedes $x'_{j+1}$. If $y_i$ precedes $y'_{j+1}$, then $R_i$ precedes $R'_{j+1}$ and so the whole sequence is $W$-good. If $y'_{j+1}$ precedes $y_i$, paths $R_i$ and $R'_{j+1}$ intersect.
We then switch at the intersection, obtaining a path satisfying the demand $(c,y'_{j+1})$ and another path from $x'_{j+1}$ to $y_i$. Compared with the path $R'_{j+1}$ from $x'_{j+1}$ to $y'_{j+1}$, we essentially move down its endpoint $y'_{j+1}$ to $y_i$. From \Cref{clm:self-config}, we can continue switching until the whole sequence becomes $W$-good.
\end{itemize}  
\item If $y_i$ precedes $x'_j$, we undo the combination of demands, leaving $R^*_i\setminus R_i$ (instead of $R^*_i$) in the first sequence and making$R_i$ satisfy the corresponding demand $(c,y_i)$. Clearly, the sequence from $R^*_i\setminus R_i$ to the end is $W$-good. We only need to additionally ensure that $R_j$ (from $x_{i-1}$ to $y_{i-1}$) precedes $R^*_i\setminus R_i$.
\begin{itemize}
\item If $x_{i-1}$ lies on the left segment, then $x_{i-1}$ precedes $c$, and $R_{i-1}=(x_{i-1},y_{i-1})$ is a corridor edge and therefore precedes $(x'_j,c)$. 
\item If $x_{i-1}$ lies on the right segment, then it precedes $y_i$. Now if $y_{i-1}$ precedes $c$, then again $R_{i-1}=(x_{i-1},y_{i-1})$ is a corridor edge and therefore precedes $R^*_i\setminus R_i$. Otherwise, $R_{i-1}$ intersects $R^*_i\setminus R_i$. We switch at the intersection, obtaining a path satisfying the demand $(x_{i-1},c)$ and another path from $x'_j$ to $y_{i-1}$. Compared with $R^*_i\setminus R_i$ which is from $x'_j$ to $c$, we essentially move the endpoint $c$ down to $y_{i-1}$. From \Cref{clm:self-config}, we can continue switching until the whole sequence becomes $W$-good.
\end{itemize} 
\end{itemize}
    
For the second sequence, consider the path $R'_{j-1}$ which connects $x'_{j-1}$ to $y'_{j-1}$. Since it precedes $R'_j$ which connects $x'_j$ to $c$, whichever endpoint of $R'_{j-1}$ lying on the left segment has to be precedes $c$, so $R'_{j-1}=(x'_{j-1},y'_{j-1})$ is an entry edge and it precedes $(x_{i+1},y_{i+1})$.

\paragraph{Case 2.2: $R^*\setminus R_i$ does not belong to any $W'$-good sequence.}
Assume that $R^*\setminus R_i$ connects $x_i$ to $c$. The current sequence is $\rset=(R_1,\ldots,R_r)$, where $R_t$ connects $x_t$ to $y_t$ for all $t\ne i$, and $R_i$ connects $c$ to $y_i$. The goal is to replace $R_i$ in the sequence with $R^*\setminus R_i$ and obtain a $W$-good sequence, with potential switching needed.

%Assume that $x_i$ precedes $y_i$. 
Consider first the path $R_{i-1}$, which connects $x_{i-1}$ to $y_{i-1}$ and precedes $R_i$. 
\begin{itemize}
\item If $x_{i-1}$ lies on the left segment, then it precedes $c$, so $R_{i-1}=(x_{i-1},y_{i-1})$ is a corridor edge therefore precedes $R^*\setminus R_i$.
\item If $x_{i-1}$ lies on the right segment, then it precedes $y_i$. In this case,
\begin{itemize}
\item if $y_{i-1}$ precedes $c$, then $R_{i-1}=(x_{i-1},y_{i-1})$ is a corridor edge therefore precedes $R^*\setminus R_i$;
\item if $c$ precedes $y_{i-1}$, then $R_{i-1}$ intersects $R_i$, and we switch the intersection between paths $R^*_i$ and $R_{i-1}$, obtaining a $x_i$-$c$-$y_{i-1}$ path and a $x_{i-1}$-$y_i$ path, with the latter satisfying the demand $(x_{i-1},y_i)$ in the associated $\dset$. See \Cref{fig: switch1} for an illustration.
\end{itemize}
\end{itemize}    
Therefore, we either make sure that all previous paths in $\rset$ are single corridor edges and therefore precede $R^*\setminus R_i$, or have at hand a $x_i$-$c$-$y_{i-1}$ path and continue with same arguments to path $R_{i-2}$, and then $R_{i-3}$ if needed, eventually ensuring that all previous paths in $\rset$ precede the $x_i$-$c$-$y_{*}$ path.

%Next, we consider $(x_{i-2},y_{i-2})$ which is above $(x_{i-1},y_{i-1})$. If $x_{i-2}$ is on the left, then it is above $y_{i-1}$, so $(x_{i-2},y_{i-2})$ is either above $(x_i,c)$ or intersects $(c,y_{i-1})$. If $x_{i-2}$ is on the right, then $y_{i-2}$ is on the left and above $y_{i-1}$ so again $(x_{i-2},y_{i-2})$ is either above $(x_i,c)$ or intersects $(c,y_{i-1})$. If it's the first case, then we are done, and if it is the later case, we switch the intersection and get $(x_i,c,y_{i-2})$, and we can do the same thing for $(x_{i-3},y_{i-3})$ and so on. In the end, we either make sure the previous edge is an entry edge that is above $(x_i,c)$ or there is no previous edge anymore. 

Consider now the path $R_{i+1}$ which connects $x_{i+1}$ to $y_{i+1}$ and is preceded by $R_i$. Note that either $x_{i+1}$ or $y_{i+1}$ is preceded by $c$, so either $R_{i+1}$ is preceded by $R^*\setminus R_i$, in which case we are done, or it intersects $R^*\setminus R_i$, in which case we switch the intersection, obtaining a $x_{i+1}$-$c$-$y_i$ path and another $x_i$-$y_{i+1}$ path satisfying the corresponding demand in $\dset$. 
We can continue with similar arguments to path $R_{i+2}, R_{i+3},\ldots$ if needed, and eventually ensuring that the $x_*$-$c$-$y_{i+1}$ path precedes all the remaining paths.
%note that $(c,y_i)$ is still in the path, which means we can remove one path from the original sequence. So $(x_{i+1},y_{i+1})$ is below $(x_i,c)$.

%Assume now that $y_i$ precedes $x_i$.

\begin{figure}[h!]
	\centering
	\subfigure[Switching paths $R^*_i$ and $R_{i-1}$ at their intersection (red), obtaining an $x_i$-$c$-$y_{i-1}$ path (black dashed line) and an $x_{i-1}$-$y_i$ path.]
	{\scalebox{0.14}{\includegraphics{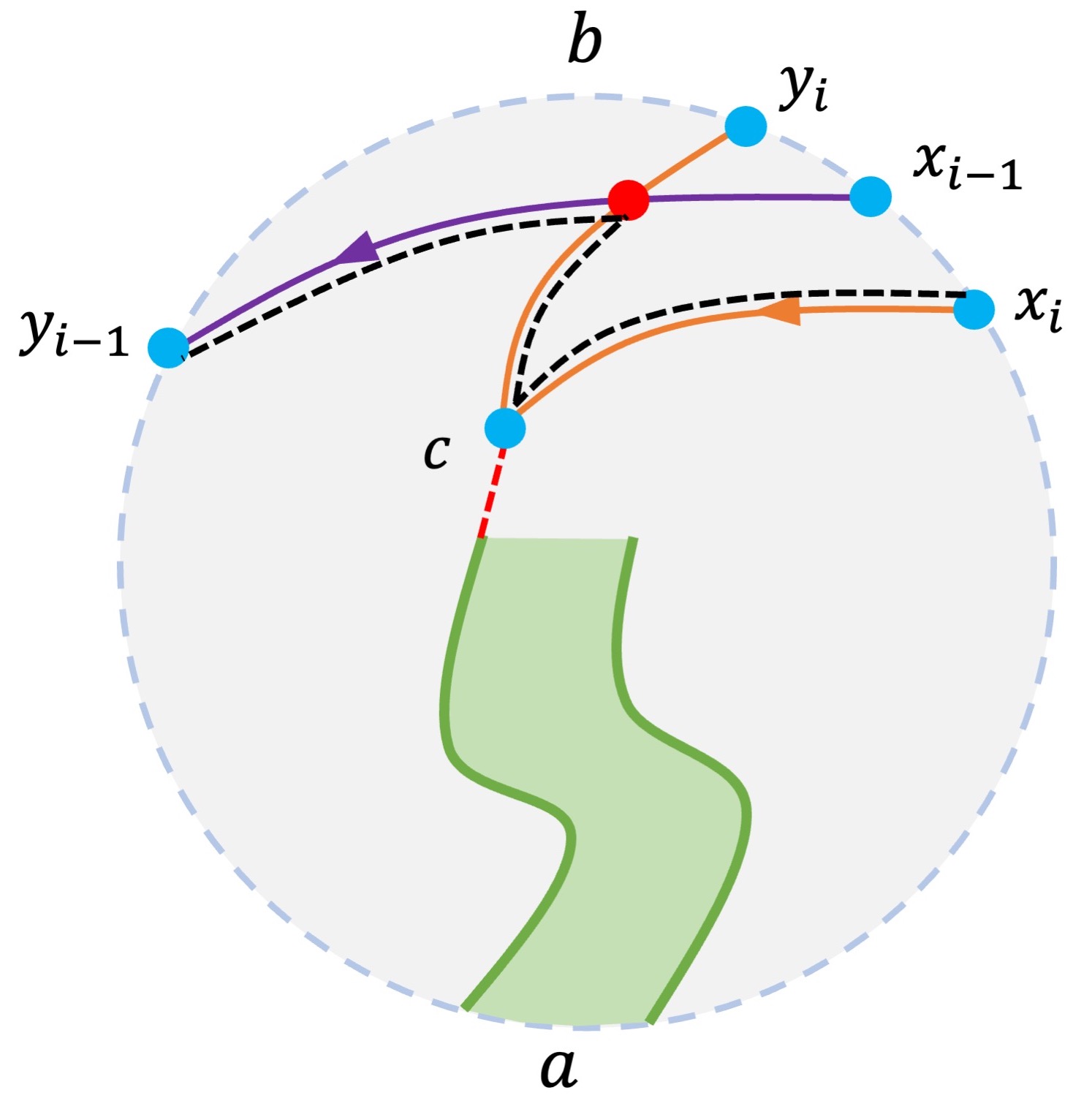}\label{fig: switch1}}}
	\hspace{1.4cm}
	\subfigure[Switching $R^*_i$ and $L^*_j$ at their intersection (red), obtaining an $x_i$-$c$-$y'_j$ path and an $x'_j$-$d$-$y_i$ path.]
	{
		\scalebox{0.14}{\includegraphics{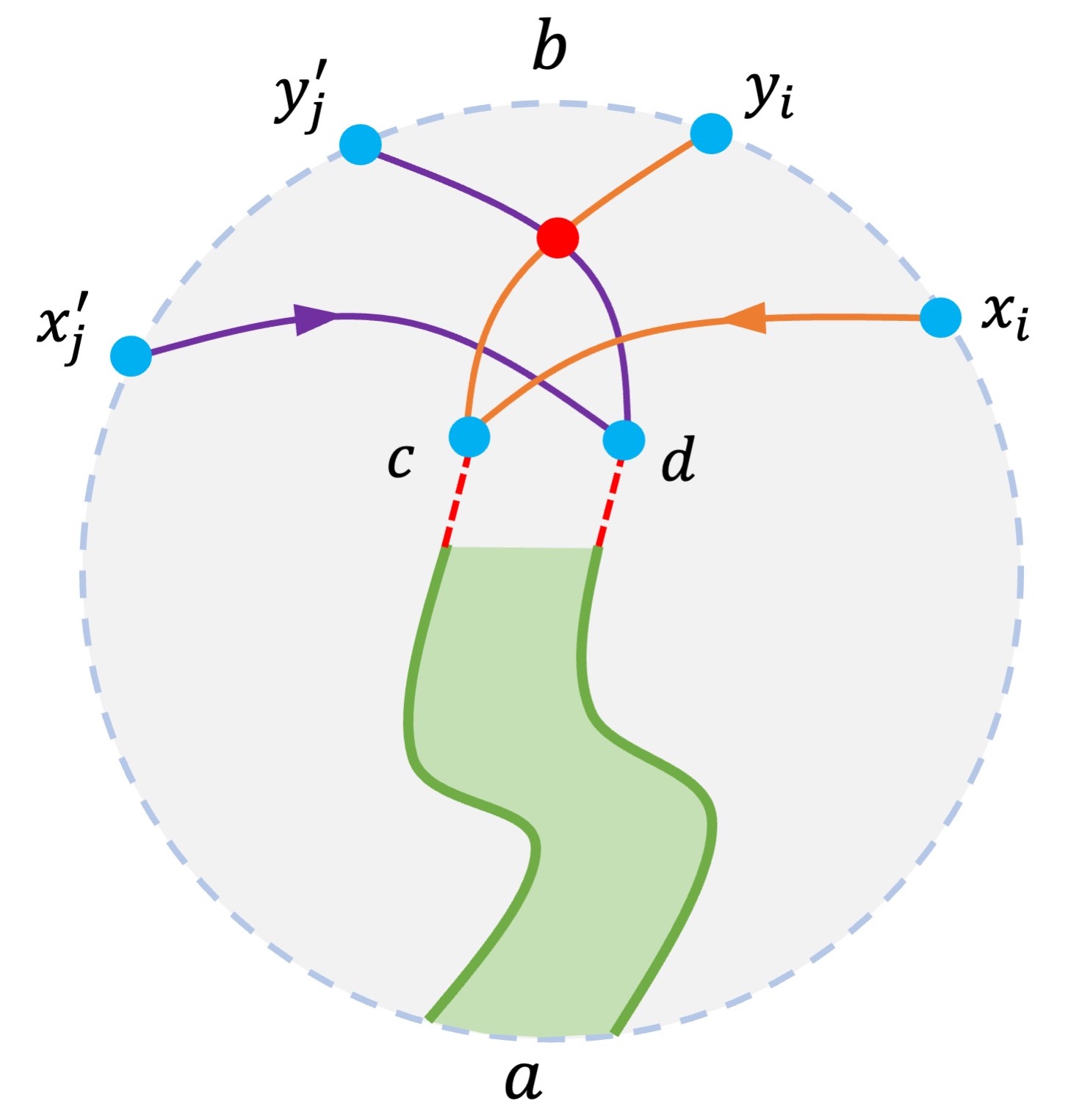}\label{fig: switch2}}}
	\caption{An illustration of switching in Case 2.2 and Case 3.\label{fig: switch}}
\end{figure}

\paragraph{Case 3: The remaining possibility.}
The above two cases dealt with active sequences in the $W'$-good tuple $(\qset',\cset',\bar\cset')$, ensuring that in the decomposition of $\qset'$, in each active $W'$-good sequence $\rset$, the first path $R_i$ containing $c$ as an endpoint satisfy that $R^*_i$ is non-crossing, starts and ends on the right segment, and contains $R_i$ as a prefix. Similar arguments apply to the $W''$-good tuple $(\qset'',\cset'',\bar\cset'')$. We can reduce to the case where in the decomposition of $\qset''$, in each active $W''$-good sequence $\lset$, the first path $L_j$ containing the right wall frontier point $d$ as an endpoint satisfy that $L^*_j$ is non-crossing, starts and ends on the left segment, and contains $L_j$ as a prefix.

\iffalse
In the end, when we consider $\gamma$, for any sequence, either it is still good as in case 1, or there is some path has the form $(x_i,c,y_i)$ where $x_i$ and $y_i$ are both on the right as in case 2. Furthermore, $(x_i,c)$ is the part in the sequence before.

Now we do the same for $\gamma_r$ and the $\qset^r$ given by the induction hypothesis, and in the end the problematic path has the form $(x_i,d,y_i)$ where $x_i$ and $y_i$ are on the left and $d$ is the new active point on the right. 
\fi

Consider now a sequence $\rset=(R_1,\ldots,R_r)$ and a sequence $\lset=(L_1,\ldots,L_\ell)$, where $R_i$ connects $x_i$ to $y_i$ and $L_j$ connects $x'_j$ to $y'_j$.
Let $c,R_i,R^*_i$ and $d,L_j,L^*_j$ be as defined above, so $R^*_i$ starts from $x_i$, passes through $c$ and ends at $y_i$, and $L^*_j$ starts from $x'_j$, passes through $d$ and ends at $y'_j$.

Clearly, subpaths $R^*_i[c,y_i]$ and $L^*_j[d,y'_j]$ intersect. We switch paths $R^*_i$ and $L^*_j$ at this intersection, obtaining an $x_i$-$c$-$y'_j$ path, which we denote by $\bar R$, and an $x'_j$-$d$-$y_i$ path, which we denote by $\bar L$. See 
\Cref{fig: switch2} for an illustration.
We now reorganize two sequences as follows:
\begin{itemize}
\item the first sequence is $(R_1,\ldots,R_{i-1},\bar L, L_{j+1},\ldots,L_{\ell})$; and 
\item the second sequence is $(L_1,\ldots,L_{j-1},\bar R, R_{i+1},\ldots,R_{r})$.
\end{itemize}

We now use \Cref{clm:self-config} to convert them into $W$-good sequences, as follows. Consider the first sequence. Note that $R_{i-1}$ (with endpoints $x_{i-1},y_{i-1}$) precedes $R_i$ (with endpoints $x_i,c$), so either $x_{i-1}$ or $y_{i-1}$ precedes $c$, so $R_{i-1}$ has to be a corridor edge and therefore precedes $\bar L$. On the other hand, note that, compared with path $L_j$ (with endpoints $x'_j,d$), the path $\bar L$ only moves the ending point from $d$ down to $y_i$, so from \Cref{clm:self-config}, we can switch to obtain a $W$-good sequence.
The second sequence can be made $W$-good in the same way.

Finally, we construct a $W$-good tuple as follows.
Let $\alpha$ be the number of such active sequences $\rset$ in $\qset'$ and $\beta$ be the number of such active sequences $\lset$ in $\qset''$.
We make $\beta$ copies of the collection $\qset'$ and $\alpha$ copies of the collection $\qset''$, and then pair each active sequence $\rset$ in these copies of $\qset'$ with a distinct active sequence $\lset$ in these copies of $\qset''$, and convert them into $W$-good sequences as above.
After all the active sequences are converted, the resulting collection is a $W$-good tuple.

\section{Reducing the size to $O(|T|^6)$}

In the previous sections, we proved that Monge property is a sufficient and necessary condition for the input quasi-metric $D$ to be realizable by a directed Okamura-Seymour instance. In this section, we show that it is even realizable by a directed Okamura-Seymour instance on $O(|T|^6)$ vertices.

Denote $k=|T|$. Recall that the Okamura-Seymour instance we use to realize the input quasi-metric is a nest $G$, a directed graph consisting of a collection $\Pi$ of edge-disjoint paths connecting pairs of terminals. We perform iterations to simplify $G$, ensuring that $G$ is a nest throughout, and in the end reducing its size to $k^6$.
In each iteration, we check whether or not there exists a pair of paths in $\Pi$ that satisfies the following conditions, and simplify $G$ accordingly.

\paragraph{Case 1. paths $P,P'$ intersect twice in the same direction.} That is, there exist a pair $x_1,x_2$ of vertices, in $G$, such that $x_1,x_2\in V(P),V(P')$, and $x_1$ appear before $x_2$ on both $P,P'$.

In this case, we exchange the subpaths of $P,P'$ between $x_1,x_2$, and then eliminate vertices $x_1,x_2$.
Edge lengths are modified accordingly.
See \Cref{fig: nudge} for an illustration.
It is easy to verify that the distance between every (ordered) pair of terminals is preserved, and so paths in $\Pi$ are still their shortest paths.
After this iteration, the number of vertices in $G$ decreases by $2$.

\begin{figure}[h!]
	\centering
	\includegraphics[scale=0.13]{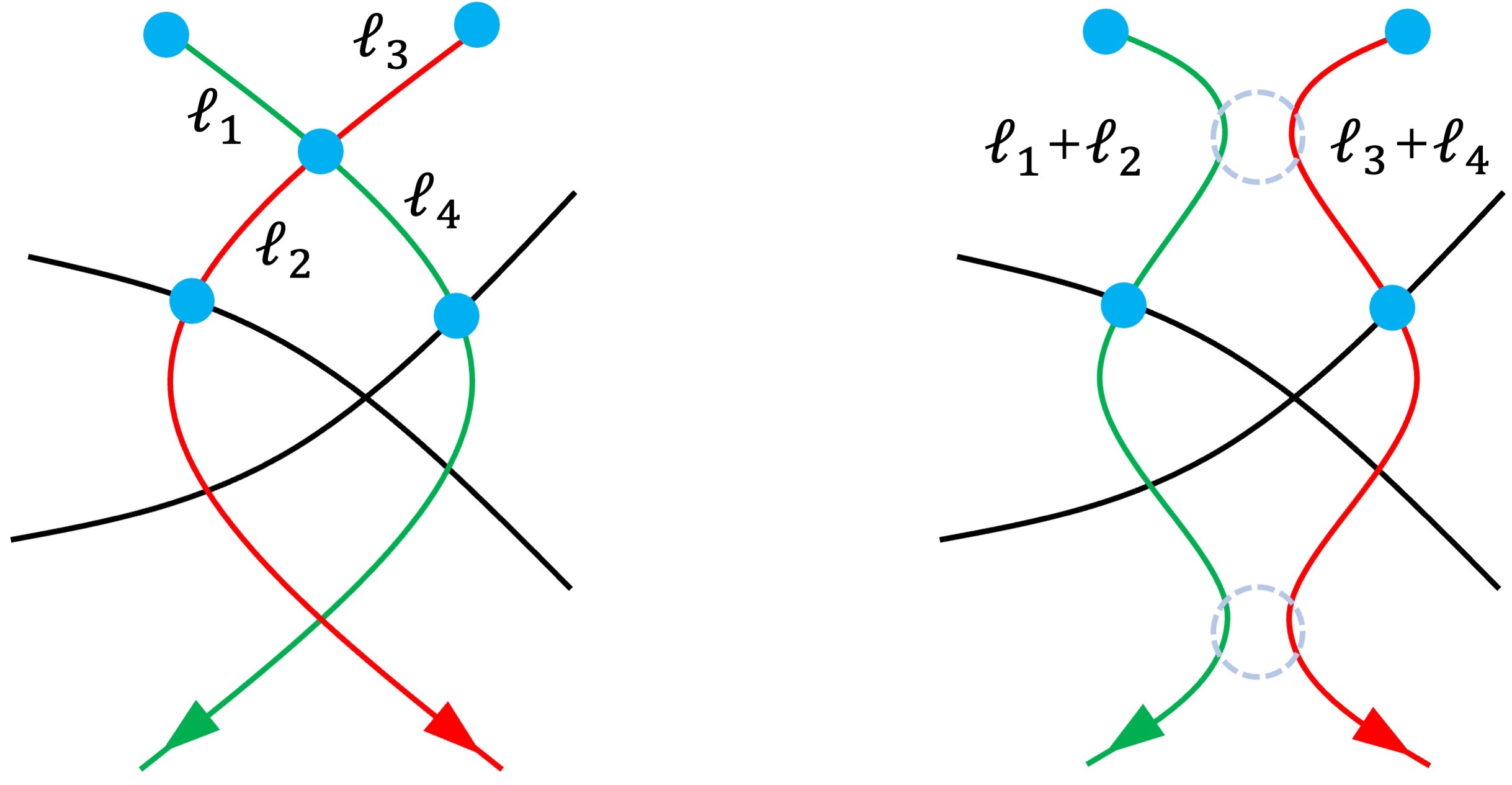}
	\caption{Exchanging the subpaths in Case 1. Edge lengths are modified accordingly.
		\label{fig: nudge}} 
\end{figure}

\paragraph{Case 2. paths $P,P'$ surrounds a $\Theta$-free area.} That is, there exist a pair $x_1,x_2\in V(P),V(P')$, such that 
\begin{itemize}
\item $x_1$ appear before $x_2$ on $P$;
\item $x_2$ appear before $x_1$ on $P'$; and
\item there is no path intersecting both $P[x_1,x_2]$ and $P'[x_2,x_1]$ (in other words, no path enters the area surrounded by $P[x_1,x_2]$ and $P'[x_2,x_1]$ from one side and exits from the other side).
\end{itemize}

In this case, we first nudge all subpaths in the area to both side until they are out of it, and then we exchange the subpaths of $P,P'$ between $x_1,x_2$, and then eliminate vertices $x_1,x_2$. Edge lengths are modified accordingly.
See \Cref{fig: nudge_2} for an illustration.
It is easy to verify that the shortest paths in $\Pi$ and their distances are preserved.
After this iteration, the number of vertices in $G$ decreases by at least $2$.

\begin{figure}[h!]
	\centering
	\includegraphics[scale=0.13]{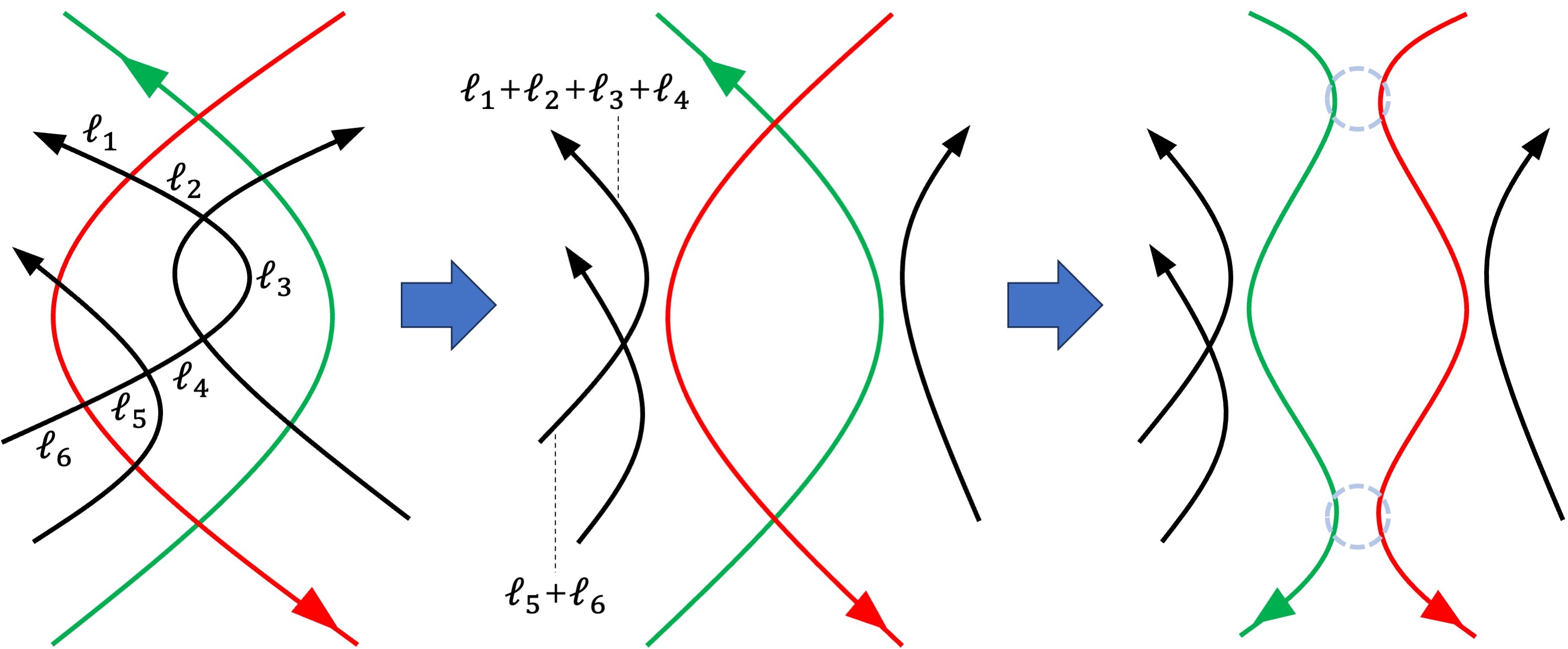}
	\caption{Nudging and exchanging subpaths in Case 2. Edge lengths are modified accordingly.
		\label{fig: nudge_2}} 
\end{figure}

Clearly, after simplification step in either Case 1 or Case 2, $G$ is still a nest.
We keep performing iterations until there are no pairs of paths in $\Pi$ satisfying the conditions of Case 1 or Case 2. We show that $G$ contains at most $O(k^6)$ vertices, which follows from the next claim (since there are $k^2$ paths in $\Pi$, $|V(G)|\le \binom{k^2}{2}\cdot O(k^2)=O(k^6)$).

\begin{claim}
In the end, every pair of paths in $\Pi$ share at most $O(k^2)$ vertices.
\end{claim}
\begin{proof}
Consider a pair $P,P'$ of paths, and let $x_1,x_2,\ldots,x_r$ be the vertices in $V(P)\cap V(P')$, that appear on $P$ in this order. Since the paths $P,P'$ do not satisfy the condition of Case 1, the order in which that they appear on $P'$ must be $x_r,x_{r-1},\ldots,x_1$.
For each index $1\le i\le r$, since the area surrounded by $P[x_i,x_{i+1}]$ and $P'[x_{i+1},x_i]$ is not $\Theta$-free, there is some path $Q_i\in \Pi$ that enters the area from one side and exits from the other side. Note that $Q_i$ may not enter and exit in the same way at any other area, since otherwise it will create two same-direction intersections with either $P$ or $P'$. Therefore, each area will lock some path in $\Pi$ for itself, so $r\le O(k^2)$.
\end{proof}

\section{Finding the Circular Ordering: Proof of \Cref{cor: algorithm}}

In this section, we give an efficient algorithm for determining if the input quasi-metric $D$ is realizable by some directed Okamura-Seymour instance. Specifically, the algorithm will decide if there exists a circular ordering $\sigma$, such that $D$ satisfies Monge property with respect to $\sigma$, and computes one if it exists.

First, some intuitions of our algorithm. Imagine that we have somehow figured out that  in the true ordering $\sigma$, terminals $a,b$ appear consecutively. We can use this knowledge to discover the locations of other terminals in $\sigma$ as follows. Consider a pair $c,d\in T$, and we compute $D(a,c)-D(b,c)$ and $D(a,d)-D(b,d)$. If $D(a,c)-D(b,c)>D(a,d)-D(b,d)$, then by simply reorganizing terms we get
$D(a,c)+D(b,d)>D(a,d)+D(b,c)$, and so $(a,d)$ and $(b,c)$ cannot cross due to Monge property, and so $a,b,c,d$ appear in this order. 
Similarly, if $D(a,c)-D(b,c)<D(a,d)-D(b,d)$, then we can derive that $a,b,d,c$ must appear in this order. Therefore, we calculate $D(a,t)-D(b,t)$ for all other terminals $t$, and if all these values are distinct, we can immediately find the true ordering $\sigma$.

The hard case is when many of the values $\set{D(a,t)-D(b,t)}_{t\in T}$ are the same, and we are incapable of determining the relative ordering on these terminals. However, as long as not all these values are equal, we have still made progress by obtaining a refined ordering, for example we may now know that $a,b,c,d$ must appear in this order, and between $d$ and $a$ there is another group of terminals. Then we can rely on $c,d$ and compute the values $\set{D(c,t)-D(d,t)}$ for terminals $t$ in this group to further refine its ordering. 
If at some point we cannot refine the ordering anymore, we will recurse on each group, try figuring out another consecutive pair (like the pair  $(a,b)$ we start with), and use this additional knowledge to further refine the ordering.

%The idea of the algorithm is the following: we start with a terminal $a$, then we guess what is the next (clockwise) terminal. For any terminal $b$, we will run a polynomial time test if $b$ could be the next terminal, if $b$ doesn't pass, then $b$ could not be the next terminal. We also prove that if $b$ does pass the test, Then we can divide the rest of the terminals into groups and we only need to determine the orders within each group, thus reduce the problem into smaller size one. 

%We now describe the test. When $b$ is the next termianl (clockwise) of $a$, image that there are two terminals $c$ and $d$, such that $D(a,c)+D(b,d)<D(a,d)+D(b,c)$, by Monge property, we know that $(a,b,c,d)$ cannot be the order on the boundary. We will try to get this kind of knowledge as many as possible, and see if they contradict with each other.

We now describe the algorithm in detail. We pick an arbitrary pair $a,b$ of terminals, and run tests on them. The goal is to decide if they can appear consecutively on the boundary.
Throughout, we maintain a circularly ordered partition $\Sigma$ of the terminals, starting with $\Sigma=(\set{b},T\setminus \set{a,b},\set{a})$ (appearing clockwise in this order), and we will iteratively refine it. The sets in $\Sigma$ are called \emph{groups}, so initially there are three groups in $\Sigma$.
%We will also think of $\Sigma$ as a circular ordering on its groups, so initially $\set{b},T\setminus \set{a,b},\set{a}$ appear clockwise in this order.

%
%For a pair of terminals $t_1, t_2$, we say that $t_1\prec t_2$ if $t_1$ cannot appear be the order on the boundary. Since $a$ and $b$ is the neighbor, $\prec$ has transitivity. At first, we do not know any relationships betweenn the termianls, we say every terminal except $a,b$ are in the same group. As the test going on, we will figure out some relationships within the group, for example $c \prec d$, we prove that then we can partition the groups into several small groups, and for any two different groups $S_1$ and $S_2$, for all $t_1 \in S_1$ and $t_2$ in $S_2$, we either always have $t_1 \prec t_2$ or we always have $t_2 \prec t_1$. And we also prove that, if we can no longer figure out any more relationship among a group, than the order within the group will not affect the Monge property that involve terminals outside the group, and thus we can focus on the orders in each group. 
%
We will ensure that, at any point, for every four terminals $t_1,t_2,t_3,t_4$ such that (i) they belong to different groups in the current $\Sigma$; and (ii) their corresponding groups $S_1,S_2,S_3,S_4$ appear in $\Sigma$ in this order, all Monge property inequalities on them are satisfied. After every iteration, we perform this test for all tuple of terminals, if some tuple does not satisfy Monge property, we terminate the algorithm and report \textsf{Fail}.

The tests will concern $3$ groups and $4$ terminals, where one group contains two terminals, denoted by $t_1,t_2$, and others groups are singletons, denoted by $\set{c}$ and $\set{d}$. 
Assume without loss of generality that $\set{c},\set{d}$, and the group containing $t_1,t_2$ appear clockwise in this order.
The goal of this test is to try to decide the clockwise order on the four terminals, which is either $c,d,t_1,t_2$ or $c,d,t_2,t_1$. 

%describe these tests and analysis them one by one. 
In all tests below, the criterion will imply that $c,d,t_2,t_1$ is the order. In other words, if we cut the circular ordering between $a$ and $b$ and obtain a linear ordering from $a$ to $b$, then $t_1$ appear before $t_2$ in this ordering (i.e., $t_1$ is closer to $a$ than $t_2$).
We also write $t_1 \prec t_2$.
In particular, the criterion means that $t_1,t_2$ should no longer be in the same group, and the old group containing them, which we denote by $S$, needs to be refined.

\subsubsection*{Test 1: $D(c,t_1)+D(d,t_2)<D(c,t_2)+D(d,t_1)$.} 
Reorganizing terms, $D(c,t_1)-D(d,t_1)<D(c,t_2)-D(d,t_2)$. For each terminal $t$ in $S$, we compute $D(c,t)-D(d,t)$, and then we sort the terminals in the increasing order of this value and use it to refine $S$ and obtain an ordered partition $(S_1,\ldots,S_r)$ of $S$. Namely, for every pair $t,t'$ of terminals in $S$,
\begin{itemize}
\item $t\in S_i, t\in S_j$ for indices $i<j$, iff $D(c,t)-D(d,t)<D(c,t')-D(d,t')$;
\item $t,t'$ lie in the same set $S_i$, iff $D(c,t)-D(d,t)=D(c,t)-D(d,t')$.
\end{itemize}
We then replace $S$ with the refined ordered partition $(S_1,\ldots,S_r)$ in $\Sigma$. This completes the refinement of $\Sigma$ in this test. Note that, if $t\in S_i, t\in S_j$ for indices $i<j$, then from  Monge property we can indeed derive that $t \prec t'$.

\subsubsection*{Test 2: $D(t_1,c)+D(t_2,d)<D(t_2,c)+D(t_1,d)$.}
Similar to Test 1, we can refine the group $S$ based on the values $\set{D(t,c)-D(t,d)}_{t\in S}$.

\subsubsection*{Test 3: $D(c,t_1)+D(t_2,d)<D(c,d)+D(t_2,t_1)$.}
In this case, we refine the group $S$ into three groups $(S_1,S',S_2)$, as follows. 
For every $t\in S$, 
\begin{itemize}
\item if $D(c,t_1)+D(t_2,t) \ge D(c,t)+D(t_2,t_1)$, then we put $t$ in $S_1$;
\item if $D(t_2,d)+D(t,t_1) \ge D(t_2,t_1)+D(t,d)$, then we put $t$ in $S_2$;
\item if none of above two inequalities holds, then we put $t$ in $S'$.
\end{itemize}
\begin{claim}
For every $t\in S_1$, $t'\in S'$ and $t''\in S_2$, $t \prec t' \prec t''$ must hold.
\end{claim}
\begin{proof}
We prove $t \prec t'$; the proof of $t' \prec t''$ is symmetric. 
Since $t\in S_1$ and $t'\in S'$, by definition,
$$D(c,t_1)+D(t_2,t) \ge D(c,t)+D(t_2,t_1), \quad\text{and}\quad D(c,t_1)+D(t_2,t') < D(t_2,t_1)+D(c,t').$$ Combined together, $$D(c,t)-D(t_2,t) \le D(c,t_1)-D(t_2,t_1) < D(c,t')-D(t_2,t'),$$ which means $D(c,t)+D(t_2,t') < D(c,t')+D(t_2,t)$. By Monge property, the pairs $(c,t)$ and $(t_2,t')$ cannot cross. On the other hand, since $D(c,t_1)+D(t_2,d)<D(c,d)+D(t_2,t_1)$, $$D(c,t)-D(t_2,t) \le D(c,t_1)-D(t_2,t_1) < D(c,d)-D(t_2,d),$$ so pairs $(c,t)$ and $(t_2,d)$ cannot cross. This means that $c,d,t_2,t$ appear in this order, and so $t \prec t_2$. Finally, since $t_1,t_2,t,t'$ lie in a group not containing $c$, and $(c,t),(t_2,t')$ cannot cross, $t \prec t'$.
\end{proof}

\subsubsection*{Test 4: $D(t_1,c)+D(d,t_2)<D(d,c)+D(t_1,t_2)$.}
Similar to Test 3, we refine the group $S$ into three groups $(S_1,S',S_2)$ by checking whether or not $D(t_1,c)+D(t,t_2) \ge D(t,c)+D(t_1,t_2)$ and $D(d,t_2)+D(t_1,t) \ge D(t_1,t_2)+D(d,t)$ hold.

We keep refining the ordered partition $\Sigma$ via the above four tests, until none of the criterion holds for any tuple $c,d,t_1,t_2$.
If the algorithm never reports \textsf{Fail} until the last moment. We return \textsf{Pass}, meaning that $b$ is allowed to be placed consecutive with $a$.

We perform the above test on all pairs $\set{(a,b)\mid b\in T}$.
If no terminal passes this test with $a$, then we report that $D$ is not realizable by any directed Okamura-Seymour instance. Otherwise, we take any $b$ that passes the test and the resulting ordered partition $\Sigma$ it produced. Then for each group $S\in \Sigma$, we recurse on the sub-quasi-metric of $D$ induced by terminals in $\set{a}\cup S$. If some of them reports not realizable, then we report that $D$ is not realizable as well. Otherwise, take the circular orderings returned from the subroutines on these groups, and combine them to obtain a circular ordering on $T$. This completes the description of the algorithm.

\paragraph{Runtime.}
Every time we refine a group, it takes $O(k^4)$ time to perform tests for all four-terminal tuples. Since we can refine at most $k$ times (before $\Sigma$ contains only singletons), the total time before we recurse on subinstances is $O(k^5)$. On the other hand, the total number of subinstances over the course of the algorithm is $O(k)$, the total running time is $O(k^6)$.

\subsection*{Analysis of the algorithm}

If the input quasi-metric $D$ is realizable, then some terminal will pass the tests with $a$, and can be found efficiently using our tests. Also, if $D$ is realizable, then the sub-quasi-metric induced on any subset of terminals is also realizable.
Therefore, it suffices to show that, given a valid ordering for each group in $\Sigma$, the ordering on $T$ obtained by combining them is a valid ordering, and
it suffices to verify Monge property for any tuple $(c,d,e,f)$ of $4$ terminals. Assume that they appear in this order in the combined ordering.

\paragraph{Case 1: $c,d,e,f$ are all from different groups.}
As we have constantly performed test for such tuple of $4$ terminals throughout the algorithm, Monge property on them is guaranteed.

\paragraph{Case 2: $c,d,e$ belong to a group $S$, while $f\notin S$.}
From our algorithm, we have created an instance on the group $S\cup \set{a}$, and get an valid ordering from it. From Monge propery on $a,c,d,e$, $D(a,d)+D(c,e) \ge D(a,e)+D(c,d)$. On the other hand, since $d$ and $e$ were not distinguishable by $a,f$, $D(a,d)-D(a,e)=D(f,d)-D(f,e)$, thus $D(f,d)+D(c,e) \ge D(f,e)+D(c,d)$. The other three inequalities can be derived in similar ways.

\paragraph{Case 3: $c,d\in S$, $e\in S'$, $f\in S''$.}
Such a tuple will go through our four tests, and the fact that they pass the test (and $c,d$ remains in the same group) implies that Monge property hold for them.

\paragraph{Case 4: $c,d\in S$, and $e,f\in S'$.}
Assume that $c,d \prec e,f$. As $c,d$ remain in the same group, 
$$D(a,c)+D(e,d)=D(a,d)+D(e,c),\quad\text{and}\quad D(a,c)+D(f,d)=D(a,d)+D(f,c),$$ 
so $D(e,d)-D(e,c)=D(a,d)-D(a,c)=D(f,d)-D(f,c)$, and $D(e,d)+D(f,c)=D(e,c)+D(f,d)$. Similarly, we can also derive that $D(d,e)+D(c,f)=D(c,e)+D(d,f)$.
On the other hand, since
$$D(a,d)+D(c,e) \ge D(c,d)+D(a,e),\quad\text{and}\quad D(f,d)+D(a,e) \ge D(a,d)+D(f,e),$$ 
so $D(c,e)-D(c,d) \ge D(a,e)-D(a,d) \ge D(f,e)-D(f,d)$, and $D(c,e)+D(f,d) \ge D(c,d)+D(f,e)$. Similarly, we can also derive that $D(e,c)+D(d,f) \ge D(d,c)+D(e,f)$. Thus all four inequalities of Monge property hold.

\paragraph{Case 5: $c,d,e,f$ all lie in the same group $S$.} Since the instance on $\set{a}\cup S$ returns a valid ordering that place $c,d,e,f$ in this order.
The Monge property on them is guaranteed.

\newpage
\bibliographystyle{alpha}
\bibliography{REF}

\end{document}